\documentclass[amsthm]{elsart}
\usepackage{yjsco}
\usepackage{natbib}
\usepackage{amsmath}
\usepackage{amsfonts}
\usepackage{amssymb}
\usepackage{mathrsfs}
\usepackage[plainpages=true,pdfpagelabels,colorlinks=true,citecolor=blue,hypertexnames=false]{hyperref}

\def\cF {\mathcal{F}}

\def\cH {\mathcal{H}}

\def\cR {\mathcal{R}}

\def\cW {\mathcal{W}}

\def\bN {\mathbb{N}}

\def\bV {\mathbb{V}}

\def\bZ {\mathbb{Z}}

\def\mfield {\left(k(x,y), \, (\xi_x, \eta_y)\right)}
\def\si {\sigma}
\def\de {\delta}

\def\spa {\operatorname{span}}

\def\den {{\operatorname{den}} }
\def\num { {\operatorname{num}}}

\def\wrt {with respect to}

\def\pa{\partial}

\newtheorem{theorem}{Theorem}[section]
\newtheorem{proposition}[theorem]{Proposition}

\newtheorem{lemma}[theorem]{Lemma}
\newtheorem{remark}[theorem]{Remark}
\newtheorem{define}[theorem]{Definition}
\newtheorem{example}[theorem]{Example}

\begin{document}

\begin{frontmatter}

\title{On the Existence of Telescopers for \\Mixed Hypergeometric Terms}

%
%
%
%
%

%
\author{Shaoshi Chen$^a$, \ Fr\'ed\'eric Chyzak$^b$}
\ead{schen@amss.ac.cn, frederic.chyzak@inria.fr}
\author{Ruyong Feng$^c$,\   Guofeng Fu$^c$,\  Ziming Li$^c$}
\ead{ryfeng@amss.ac.cn, fuguofeng@mmrc.iss.ac.cn, zmli@mmrc.iss.ac.cn}

\medskip
\address{$^a$Department of Mathematics, North Carolina State University, 27695-8205 Raleigh, USA}
\address{$^b$INRIA, 91120 Palaiseau, France}
\address{$^c$Key Laboratory of Math.-Mech., Chinese Academy of Sciences, 100190 Beijing, China}

\thanks{In this work, S.\ Chen was supported by
the NSF grant CCF-1017217.
F.\ Chyzak was supported in part by the Microsoft Research-INRIA Joint Centre.
G.\ Fu, R.\ Feng and Z.\ Li were supported in part by two grants of NSFC No.\ 60821002/F02 and No.\ 10901156.
}

\begin{abstract}
We present a criterion for the existence of telescopers for
mixed hypergeometric terms, which is based on multiplicative and additive decompositions.
The criterion enables us
to determine the termination of Zeilberger's algorithms for
mixed hypergeometric inputs.
\end{abstract}


\end{frontmatter}

\section{Introduction}\label{SECT:intro}

Given a sum~$U_n := \sum_{k=0}^\infty u_{n,k}$ to be computed,
\emph{creative telescoping\/} is a process that determines a
recurrence in~$n$ satisfied by the univariate sequence~$U=(U_n)$ from
a system of recurrences in $n$ and~$k$ satisfied by the bivariate
summand~$u=(u_{n,k})$.  A natural counterpart exists for integration.
Algorithmic research on this topic has been initiated by Zeilberger in
the early 1980s, leading in the 1990s to algorithms for summands and
integrands described by first-order equations, that is for
hypergeometric and hyperexponential creative telescoping
\citep{Zeilberger1990c,Zeilberger1991,Almkvist1990}.

The termination problem of Zeilberger's algorithms has been
extensively studied in the last two decades
\citep{WilfZeilberger1992,AbramovLe2002,Abramov2003,ChenHouMu2005} and
can be related to existence problems for other operations, like the
computation of diagonals \citep{Lipshitz1988}.
The main output of creative telescoping is the recurrence on the
sum~$U$.  It is called a \emph{telescoper\/} of~$u$.
Zeilberger's algorithms terminate if and only if telescopers exist, whence the
interest to discuss their existence.  \cite{Zeilberger1990} shows that
holonomicity, a notion borrowed from the theory of D-modules,
implies the existence of telescopers. In particular,
\cite{Wilf1992} present an elementary proof, based on the ideas
of~\cite{Fasenmyer1947} and~\cite{Verbaeten1974}, that telescopers
always exist for proper hypergeometric terms. However, holonomicity is
merely a sufficient condition, i.e., there are cases in which the
input functions are not holonomic but Zeilberger's algorithms still
terminate, see~\cite{Chyzak2009}. Therefore, a challenging problem is
to find a necessary and sufficient condition that enables us to
determine the existence of telescopers.

In view of the theoretical difficulty, special attention has been
focused on the subclass of hypergeometric terms, hyperexponential
functions, and mixed hypergeometric terms (see the definition in
Section~\ref{SUBSECT:mht}). In the continuous case, the results
by~\cite{Bernstein1971}, \cite{Kashiwara1978}, \cite{Lipshitz1988}
and~\cite{Takayama1992} show that every hyperexponential function
has a telescoper. This implies that Zeilberger's algorithms always
succeed on hyperexponential inputs. However, the situation in other cases
turns out to be more involved. In the discrete case, the first
complete solution to the termination problem has been given
by~\cite{Le2001} and~\cite{AbramovLe2002}, by deciding whether
telescopers exist for a given bivariate rational sequence in the
($q$)-discrete variables~$y_1$ and~$y_2$. According to their criterion,
the rational sequence
\[f = \frac{1}{y_1^2+y_2^2}\]
has no telescoper.
The criterion has been extended to the general case of bivariate hypergeometric
terms by~\cite{Abramov2002b, Abramov2003}. He proved that a
hypergeometric term can be written as a sum of a hypergeometric-summable term
and a proper one if it has a telescoper, see~\cite[Theorem 10]{Abramov2003}. Similar results have been
obtained in the $q$-discrete case by~\cite{ChenHouMu2005}.

\citet{Almkvist1990} presented a continuous-discrete analogue of creative telescoping.
This analogue is useful in the study of orthogonal
polynomials~\cite[Chapters~10--13]{Koepf1998}.
In analogy with the discrete case, not all mixed hypergeometric terms
have telescopers. Therefore, an Abramov-like criterion is also needed in the mixed case.

In order to unify the various cases of mixed rational terms,
\cite{ChenSinger2012} recently presented a criterion that is based on
residues analysis for the existence of telescopers for bivariate
rational functions.
In the present paper, we give a criterion, Theorem~\ref{TH:criterion}, on the existence of telescopers
for mixed hypergeometric terms, including continuous-discrete, continuous-$q$-discrete and discrete-$q$-discrete
terms. Beside the termination problem of creative telescoping,
the criterion,  together with the
results in~\citep{Hardouin2008, Schneider2010},  can be used to determine
if indefinite sums and integrals satisfy (possibly nonlinear) differential
equations (see Example~\ref{EX:transcendental}).

The rest of this paper is organized as follows. An
algebraic setting for mixed hypergeometric terms
is described in Section~\ref{SECT:prelim}, and the existence problem of telescopers
is stated in Section~\ref{SECT:telescoper}.
We define the notion of exact terms, and introduce the additive decompositions
in Section~\ref{SECT:ad}.
The criterion for the existence of telescopers for mixed hypergeometric terms
is presented in Section~\ref{SECT:criteria}.  Based on the criterion, we develop
an algorithm for deciding the existence of telescopers,
and present a few examples in Section~\ref{SECT:algo}.

\section{Preliminaries}\label{SECT:prelim}

Throughout the paper, we let~$k$ be an algebraically closed field of characteristic zero, and~$q$ be a nonzero
element of~$k$. Assume further that~$q$ is not a root of unity. Let~$k(x, y)$ be the field
of rational functions in~$x$ and~$y$ over~$k$. For a nonzero element~$f \in k(x, y)$, the denominator
and numerator of~$f$ are denoted by~$\den(f)$ and~$\num(f)$, respectively.
They are two coprime polynomials in~$k[x, y]$. For a ring~$R$, $R^{*}$ stands for~$R \setminus \{0\}$.
The symbol~$\bN$ stands for the set of nonnegative integers.

This section contains five subsections. In \S\ref{SUBSECT:field}, we
describe a field that will serve as ground field in our subsequent
algebraic constructions, and we define a (noncommutative) ring
of Ore polynomials. In \S\ref{SUBSECT:mixedsys}, we define
the notion of mixed hypergeometric terms, and describe a (commutative) ring extension of the ground field.
The ring extension contains mixed hypergeometric terms that occur in the study of existence of telescopers.
In \S\ref{SUBSECT:compatible}, we recall from \citep{ChenFengFuLi2011}
the structure of compatible rational functions, which
leads to a multiplicative decomposition of mixed hypergeometric terms given
in~\S\ref{SUBSECT:mht}.
Finally, we define in~\S\ref{SUBSECT:split} the notion of split
rational functions and the notions of proper terms that are
meaningful in presence of two types of operators.

\subsection{Fields endowed with a pair of operators} \label{SUBSECT:field}
Let~$\de_x=\pa/\pa_x$ and~$\de_y=\pa/\pa_y$ be
the usual derivations with respect to~$x$ and~$y$, respectively. For an element~$f\in k(x, y)$, we define the shift
operators~$\si_x$ and~$\si_y$ as
\[\si_x(f(x, y)) = f(x+1, y)\quad \text{and}\quad \si_y(f(x, y)) = f(x, y+1),\]
and~$q$-shift operators~$\tau_x$ and~$\tau_y$ as
\[\tau_x(f(x, y)) = f(qx, y)\quad \text{and}\quad \tau_y(f(x, y)) = f(x, qy).\]

To describe the mixed cases concisely, we introduce the following notation.

\medskip \noindent
{\bf Notation.}
Let~$\Theta$ denote the set
\[\{\de_x, \si_x, \tau_x\} \times \{\de_y, \si_y, \tau_y\} \setminus \{(\de_x, \de_y), (\si_x, \si_y), (\tau_x, \tau_y)\}.\]
A pair~$(\xi_x, \eta_y)\in \Theta$ is called a~\emph{mixed\/} pair of operators.

\medskip
Note that, for every~$(\xi_x, \eta_y) \in \Theta$, $\xi_x$ and~$\eta_y$ act on variables of different types.
It follows that~$\xi_x \circ \eta_y (f) = \eta_y \circ \xi_x(f)$ for all~$f \in k(x,y)$.

In the sequel,~$k(x,y)$ is usually endowed with a mixed pair~$(\xi_x, \eta_y)$ of operators.
The resulting structure is denoted
as~$\left(k(x,y), \, (\xi_x, \eta_y) \right)$. There are six cases listed in Figure~\ref{FIG:fields}, in which
D-$\Delta$ stands for the differential and difference case, D-$\Delta_q$ for the differential and $q$-difference case, etc.
Of course, the last three cases can be identified with the first three when we swap~$x$ and~$y$ in the field~$k(x,y)$.
\begin{center}
\begin{figure}
\begin{center}
\begin{tabular}{|c|l|c|}
\hline
        Case                 & $(\xi_x, \eta_y)$ &  Field endowed with~$(\xi_x, \eta_y)$ \\    \hline
D-$\Delta$      & $(\de_x,\si_y)$        & $\left(k(x, y), (\delta_x, \si_y)\right)$ \\
D-$\Delta_q$  & $(\de_x,\tau_y)$       &  $\left(k(x, y), (\delta_x, \tau_y)\right)$ \\
$\Delta$-$\Delta_q$    & $(\si_x,\tau_y)$       & $\left(k(x, y), (\si_x, \tau_y)\right)$  \\ \hline
$\Delta$-D      & $(\si_x,\de_y)$        & $\left(k(x, y), (\si_x, \de_y)\right)$ \\
$\Delta_q$-D  & $(\tau_x,\de_y)$       &  $\left(k(x, y), (\tau_x, \de_y)\right)$ \\
$\Delta_q$-$\Delta$    & $(\tau_x,\si_y)$       & $\left(k(x, y), (\tau_x, \si_y)\right)$  \\
\hline
\end{tabular}
\end{center}
\medskip
\caption{Fields endowed with a mixed pair of operators} \label{FIG:fields}
\end{figure}
\end{center}

Given a field~$\left(k(x,y), (\xi_x, \eta_y)\right)$,
one can define a ring of Ore polynomials \citep{ChyzakSalvy1998},
which we denote here by~$k(x,y)\langle \pa_x, \pa_y \rangle$.
According to the different choices of~$(\xi_x, \eta_y)$, we illustrate the corresponding rings of
Ore polynomials in Figure~\ref{FIG:rings}.
In particular, the commutation rules for the explicit case~$f = x$ are
as follows:
\begin{equation*}
D_x x = x D_x + 1, \
S_x x = (x + 1) S_x, \
T_x x = q x T_x, \
D_y x = x D_y, \
S_y x = x S_y, \
T_y x = x T_y.
\end{equation*}

\begin{center}
\begin{figure}
\begin{center}
\begin{tabular}{|c|l|c|c|}
\hline
        Case                 & $(\xi_x, \eta_y)$ & Ring $k(x,y)\langle \pa_x, \pa_y \rangle$  & Commutation rules, where~$f \in k(x,y)$ \\    \hline
D-$\Delta$      & $(\de_x,\si_y)$        & $k(x, y) \langle D_x, S_y\rangle$ &  $D_x f = f D_x + \delta_x(f)$
and $S_y f = \si_y(f) S_y$  \\
D-$\Delta_q$  & $(\de_x,\tau_y)$       & $ k(x, y)\langle D_x, T_y\rangle$ &  $D_x f = f D_x + \delta_x(f)$ and $T_y f = \tau_y(f) T_y$  \\
$\Delta$-$\Delta_q$    & $(\si_x,\tau_y)$       & $k(x, y)\langle S_x, T_y\rangle$ & $S_x f = \si_x(f) S_x$ and $T_y f = \tau_y(f) T_y$   \\ \hline
$\Delta$-D      & $(\si_x, \de_y)$        & $k(x, y) \langle S_x, D_y\rangle$ &   $S_x f = \si_x(f) S_x $ and  $D_y f = f D_y + \delta_y(f)$ \\
$\Delta_q$-D  & $(\tau_x, \de_y)$       & $ k(x, y)\langle T_x, D_y\rangle$ &  $T_x f = \tau_x(f) T_x$ and $D_y f = f D_y + \delta_y(f)$   \\
$\Delta_q$-$\Delta$    & $(\tau_x, \si_y)$       & $k(x, y)\langle T_x, S_y\rangle$ & $T_x f = \tau_x(f) T_x$ and  $S_y f = \si_y(f) S_y$    \\
\hline
\end{tabular}
\end{center}
\caption{Rings of Ore polynomials} \label{FIG:rings}
\end{figure}
\end{center}

\subsection{First-order mixed systems} \label{SUBSECT:mixedsys}
A first-order mixed linear-functional system is of the form
\begin{equation} \label{EQ:mixedsys}
\left\{ \begin{array}{rcl}
\xi_x(z) & = a z, \\
\eta_y(z) & = b z,
\end{array} \right.
\end{equation}
where~$(\xi_x, \eta_y) \in \Theta$ and~$a, b \in k(x,y)$.
For brevity, we call~\eqref{EQ:mixedsys} a {\em first-order mixed system\/} or a {\em mixed system\/} in the sequel.
\begin{example} \label{EX:zerodivisor}
Let~$(\xi_x, \eta_y)=(\delta_x, \si_y)$, $a=y/x$ and~$b=-x$. The system~\eqref{EQ:mixedsys} becomes
\[
\left\{ \begin{array}{rcl}
\delta_x(z) & = \frac{y}{x} \, z,  \\
\si_y(z) & = -x \,  z.
\end{array} \right.
\]
It is straightforward to verify that the expression $(-x)^y$ solves this mixed system.
Moreover, this system does not have any nonzero rational solution in~$k(x, y)$: if
it had, we could write such a solution in the form
\[ z =\frac{P}{Q} = \frac{p_my^m + \cdots + p_0}{q_n y^n + \cdots + q_0}, \quad \text{where~$p_i$ and~$q_j$ are in~$k(x)$ with~$p_mq_n\neq 0$.}\]
By the equality~$\si_y(z) = -x z$, we have~$\si_y(P)Q = -x P \si_y(Q)$. Equating the leading coefficients with respect
to~$y$ yields~$p_m q_n = -xp_mq_n$, which further implies that~$x=-1$. This is a contradiction with the assumption that~$x$ is
transcendental over~$k$.
\end{example}

The example above shows that solving generally requires to extend the
filed~$k(x,y)$.  This motivates us to consider ring extensions
of~$k(x,y)$ endowed with a mixed pair of operators.
\begin{define} \label{DEF:ringext}
For a pair~$(\xi_x, \eta_y) \in \Theta$,
we call a triple~$\left(R, (\bar{\xi}_x, \bar{\eta}_y) \right)$ a {\em ring extension\/}
of~$\left(k(x,y), (\xi_x, \eta_y)\right)$ if the following conditions are satisfied.
\begin{enumerate}
\item[(i)] $R$ is a commutative ring containing~$k(x,y)$.
\item[(ii)] $\bar{\xi}_x: R \rightarrow R$ is an extension of~$\xi_x$, and
$\bar{\eta}_y: R \rightarrow R$ is an extension of~$\eta_y$.
\item[(iii)] $\bar{\xi}_x$ is a derivation on~$R$ if~$\xi_x =\de_x$, and~it is a monomorphism
if~$\xi_x=\si_x$ or~$\xi_x = \tau_x$.
\item[(iv)] $\bar{\eta}_y$ is a derivation on~$R$ if~$\eta_y =\de_y$, and~it is a monomorphism
if~$\eta_y=\si_y$ or~$\eta_y = \tau_y$.
\item[(v)] $\bar{\xi}_x$ and~$\bar{\eta}_y$ commute.
\end{enumerate}
Moreover, such a  ring extension is said to be {\em simple\/} if there does not exist any
ideal~$I$ of~$R$ such that~$\bar{\xi}_x(I) \subset I$ and~$\bar{\eta}_y(I) \subset I$ except for~$I=R$ and~$I=\{0\}$.
\end{define}

Without any possible ambiguity, we shall denote the operators
$\bar{\xi}_x$ and~$\bar{\eta}_y$ obtained as in the definition above
by~$\xi_x$ and~$\eta_y$, respectively.
The reader may find more general ring extensions endowed with derivations, shift and $q$-shift operators
in~\citep{Hardouin2008}.

Let~$L=\sum_{i,j} a_{i,j} \pa_x^i \pa_y^j$ be an Ore polynomial in~$k(x,y)\langle \pa_x, \pa_y\rangle$,
where~$k(x,y)$ is endowed with a mixed pair~$(\xi_x, \eta_y)$ of operators.
Let~$(R, (\xi_x, \eta_y))$ be a ring extension of~$(k(x,y), (\xi_x, \eta_y))$.
We define the application of~$L$ to an element~$r \in R$ to be
\[ L(r) = \sum_{i,j}  a_{i,j} \xi_x^i \circ \eta_y^j (r). \]
It is straightforward to verify that, for~$L_1, L_2 \in k(x,y)\langle \pa_x, \pa_y\rangle$,
\[ L_1 (L_2(r)) = (L_1 L_2)(r) \quad \mbox{for all~$r \in R.$} \]

We are about to define the constants of a given field~$\mfield$, by
describing them in a uniform way as the solutions of specific
operators.  By application of elements in~$k(x,y)\langle \pa_x, \pa_y
\rangle$, we have
\begin{equation}\label{DEF:delta}
\Delta_y := \pa_y - \pa_y(1) = \left\{ \begin{array}{ll}
                          D_y & \mbox{if $\eta_y = \delta_y$,} \\
                          S_y - 1 & \mbox{if $\eta_y = \si_y$,} \\
                          T_y - 1 & \mbox{if $\eta_y = \tau_y$.} \\
                          \end{array} \right.
\end{equation}

The same holds when we replace~$y$ by~$x$ and~$\eta_y$ by~$\xi_x$, and
we define~$\Delta_x$ similarly.
An element~$c \in R$ is then called a {\em constant\/} with respect
to the pair~$(\xi_x, \eta_y)$ if
\[    \Delta_x(c) = \Delta_y(c) = 0. \]
One can easily verify that~$c \in k(x,y)$ is a constant with respect to~$(\xi_x, \eta_y)$ if and only if~$c$ is an element of~$k$.

Given a system of the form~\eqref{EQ:mixedsys}, a basic question is
whether there exists a ring extension~$(R, (\xi_x, \eta_y))$
containing a nonzero solution of the system.  This question is related
to compatibility conditions of~\eqref{EQ:mixedsys}, which we discuss
in the next section.

\subsection{Compatible rational functions} \label{SUBSECT:compatible}
Let~$k(x,y)$ be endowed with a mixed pair~$(\xi_x, \eta_y)$ of operators.
If the first-order mixed system~\eqref{EQ:mixedsys} has a nonzero solution~$h$ in a ring
extension~$(R, (\xi_x, \eta_x))$, by the commutativity of~$\xi_x$ and~$\eta_y$, we
have~$\eta_y \circ \xi_x(h)= \xi_x \circ \eta_y(h)$. It follows from~\eqref{EQ:mixedsys} that
\[
\eta_y(a h ) = \xi_x(b h), \quad
a \neq 0 \,\, \mbox{if~$\xi_x \neq \delta_x$}, \,\, \text{and} \,\,
b \neq 0 \,\, \mbox{if~$\eta_y \neq \de_y$.}
\]

The above three conditions lead to six pairs of compatibility
conditions listed in Figure~\ref{FIG:compatible}.  These are necessary
conditions for System~\eqref{EQ:mixedsys} to possess a solution.

\begin{center}
\begin{figure}
\begin{center}
\begin{tabular}{|c|l|c|}
\hline
        Case                 & $(\xi_x, \eta_y)$ &  Compatibility condition \\    \hline
D-$\Delta$      & $(\de_x,\si_y)$        & $\frac{\de_x(b)}{b}= \si_y(a)-a$ and $b \neq 0$ \\
D-$\Delta_q$  & $(\de_x,\tau_y)$       & $\frac{\de_x(b)}{b}= \tau_y(a)-a$ and $b \neq 0$ \\
$\Delta$-$\Delta_q$    & $(\si_x,\tau_y)$       & $\frac{\si_x(b)}{b} = \frac{\tau_y(a)}{a}$ and $ab \neq 0$  \\ \hline
$\Delta$-D      & $(\si_x,\de_y)$        & $\si_x(b)-b = \frac{\de_y(a)}{a}$ and  $a \neq 0$ \\
$\Delta_q$-D  & $(\tau_x,\de_y)$       &  $\tau_x(b)-b = \frac{\de_y(a)}{a}$ and $a \neq 0$ \\
$\Delta_q$-$\Delta$    & $(\tau_x,\si_y)$       & $\frac{\tau_x(b)}{b} = \frac{\si_y(a)}{a}$ and  $ab \neq 0$ \\
\hline
\end{tabular}
\end{center}
\caption{Compatibility conditions} \label{FIG:compatible}
\end{figure}
\end{center}

\begin{center}
\begin{figure}
\begin{center}
\begin{tabular}{|c|l|l|}
\hline
        Case                 & $(\xi_x, \eta_y)$ &  $a, b  \in k(x,y)$ compatible w.r.t.~$(\xi_x, \eta_y)$  \\    \hline
D-$\Delta$      & $(\de_x,\si_y)$        &
$\begin{array}{c}
\text{There exist~$f \in k(x,y)$, $\alpha, \beta \in k(x)$, and~$\gamma \in k(y)$ such that} \\
a = \frac{\de_x(f)}{f} + y \frac{\de_x(\beta)}{\beta} + \alpha \quad \text{and} \quad b= \frac{\si_y(f)}{f} \cdot \beta \cdot \gamma.
\end{array}$
\\
D-$\Delta_q$  & $(\de_x,\tau_y)$       &
$\begin{array}{c}
\text{There exist~$f \in k(x,y)$, $\alpha \in k(x)$, and~$\beta \in k(y)$ such that} \\
a = \frac{\de_x(f)}{f} + \alpha \quad \text{and} \quad b= \frac{\tau_y(f)}{f} \cdot \beta.
\end{array}$
\\
$\Delta$-$\Delta_q$    & $ (\si_x,\tau_y)$       &
$\begin{array}{c}
\text{There exist~$f \in k(x,y)$, $\alpha \in k(x)$, and~$\beta \in k(y)$ such that}\\
a = \frac{\si_x(f)}{f} \cdot \alpha \quad \text{and} \quad b= \frac{\tau_y(f)}{f} \cdot \beta.
\end{array}
$
\\
\hline
$\Delta$-D      & $(\si_x,\de_y)$        &
$\begin{array}{c}
\text{There exist~$f \in k(x,y)$, $\alpha, \beta \in k(y)$, and~$\gamma \in k(x)$ such that} \\
a = \frac{\si_x(f)}{f} \cdot \beta \cdot \gamma \quad \text{and} \quad b= \frac{\de_y(f)}{f} + x \frac{\de_y(\beta)}{\beta} + \alpha.
\end{array}$
\\
$\Delta_q$-D  & $(\tau_x,\de_y)$       &
$\begin{array}{c}
\text{There exist~$f \in k(x,y)$, $\alpha \in k(y)$, and~$\beta \in k(x)$ such that} \\
a = \frac{\tau_x(f)}{f} \cdot \beta \quad \text{and} \quad b= \frac{\de_y(f)}{f} + \alpha.
\end{array}$
\\
$\Delta_q$-$\Delta$    & $(\tau_x,\si_y)$       &
$\begin{array}{c}
\text{There exist~$f \in k(x,y)$, $\alpha \in k(y)$, and~$\beta \in k(x)$ such that} \\
a = \frac{\tau_x(f)}{f} \cdot \beta \quad \text{and} \quad b= \frac{\si_y(f)}{f} \cdot \alpha.
\end{array}$
\\
\hline
\end{tabular}
\end{center}
\caption{Structure of compatible rational functions} \label{FIG:structure}
\end{figure}
\end{center}

\begin{define} \label{DEF:comprat}
Let~$a, b \in k(x,y) \times k(x,y)$ and~$(\xi_x, \eta_y) \in \Theta$. We say
that~$a$ and~$b$ are {\em compatible\/} with respect to~$(\xi_x, \eta_y)$ if
the compatibility conditions corresponding to~$(\xi_x, \eta_y)$ in Figure~\ref{FIG:compatible}
are satisfied.
\end{define}

Theorem~1 in~\citep{ChenFengFuLi2011} describes the special structure of compatible rational functions.
 Applying this theorem to the six cases in Figure~\ref{FIG:compatible},
we obtain Figure~\ref{FIG:structure}, which describes the structure of compatible bivariate rational functions.
In fact, the conclusions in Figure~\ref{FIG:structure} can also be derived from Lemmas~3.1, 3.2 and 3.3
in~\citep{ChenFengFuLi2011}, because we are only concerned with bivariate compatible rational functions.

A first-order mixed system of the form~\eqref{EQ:mixedsys} is said to be {\em compatible\/}
if its coefficients~$a$ and~$b$ are compatible with respect to~$(\xi_x, \eta_y)$.
A first-order mixed system must be compatible with respect to~$(\xi_x, \eta_y)$ if
it has a nonzero solution in some ring extension of~$\left(k(x,y), (\xi_x, \eta_y) \right).$

\subsection{Mixed hypergeometric terms} \label{SUBSECT:mht}

Hypergeometric terms are a common abstraction of geometric terms and
factorials.  They play an important role in combinatorics. The
continuous analogue of hypergeometric terms is hyperexponential
functions: they generalize usual exponential functions and simple
radicals. In this paper, we will consider a class of functions in~$x$
and~$y$ that are solutions of first-order mixed systems, and are
therefore intermediate objects between hypergeometric terms and
hyperexponential functions.

Given a compatible mixed system of the form~\eqref{EQ:mixedsys},
Theorem 2 in~\citep{Bronstein2005} implies that
there exists a simple ring extension~$\left(\cR, (\xi_x, \eta_y)\right)$ of~$\left(k(x, y), (\xi_x, \eta_y) \right)$ containing a nonzero solution
of~\eqref{EQ:mixedsys}.
Such a simple ring is called a \emph{Picard--Vessiot extension\/}
associated to~\eqref{EQ:mixedsys}.
Moreover, the set of constants in~$\cR$ is equal to~$k$ due to the assumption that~$k$ is algebraically
closed. These facts allow us to define the notion of mixed hypergeometric terms in a ring setting as follows.
\begin{define} \label{DEF:mhyper}
Let~$k(x,y)$ be a field endowed with a mixed pair~$(\xi_x, \eta_y)$ of operators.
Assume that~$\left(\cR, (\xi_x, \eta_y)\right)$ is a simple ring extension of~$\left(k(x,y), (\xi_x, \eta_y)\right)$,
and that the set of constants in~$\cR$ is equal to~$k$.  A \emph{nonzero} element~$h$ of~$\cR$
is called a {\em mixed hypergeometric term\/} over~$\mfield$ if there exist~$a, b \in k(x,y)$ such that
$$\xi_x(h)=ah \quad  \text{and} \quad \eta_y(h)=bh.$$
We call~$a$ the {\em certificate of\/~$h$ with respect to\/~$\xi_x$}, and~$b$ the {\em certificate with respect to\/~$\eta_y$}.
\end{define}
For brevity, a mixed hypergeometric term will be called a mixed term in the sequel.

%
\begin{lemma} \label{LM:mhyper} Let the ring extension~$\left(\cR, (\xi_x, \eta_y)\right)$ be given as in Definition~\ref{DEF:mhyper}.
\begin{enumerate}
\item[(i)] Every mixed term is invertible.
\item[(ii)] If two mixed terms have the same certificates, then their ratio belongs to~$k$.
\end{enumerate}
\end{lemma}
\begin{proof}
Let~$h$ be a mixed term in~$\cR$, and~$I$ be the ideal generated by~$h$ in~$\cR$.
Then~$\xi_x(h)$ and~$\eta_y(h)$ belong to~$I$. It follows that~$\xi_x(I) \subset I$ and~$\eta_y(I) \subset I$.
Since~$\cR$ is simple and~$h$ is nonzero, $I=\cR$, that is, $1 \in I$. Consequently, $h$ is invertible.
The first assertion holds.

Let~$h_1$ and~$h_2$ be two mixed terms in~$\cR$. If they have the same certificates, then~$h_1/h_2$ is a constant
by a straightforward calculation, that is to say, $h_1 = c h_2$ for some~$c \in k$.
The second assertion holds.
\end{proof}

Viewing mixed terms in an abstract ring allows us to compute their sums, products and inverses legitimately.
Moreover, we will never encounter any analytic considerations, such as singularities and the regions of definition.
This choice will
not do any harm, as the problem we are dealing with is purely algebraic.

Two mixed terms~$h_1$ and~$h_2$ are said to be~{\emph{similar\/}} if the ratio~$h_1/h_2$ is in~$k(x, y)$.
It is easy to verify that similarity is an equivalence relation.

When studying the existence of telescopers, we will encounter at most
finitely many mixed terms that are dissimilar to each other. These
terms can be regarded as elements in a simple ring, because a finite
number of Picard--Vessiot extensions associated to compatible
first-order mixed systems can be embedded into a simple ring \citep[\S
  2.2]{LiSingerWuZheng2006}.
From now on, we assume that~$\cR$ is given as in Definition~\ref{DEF:mhyper}.
It will be sufficient to consider mixed terms in~$\cR$.

By the second assertion of Lemma~\ref{LM:mhyper}, two mixed terms having the same certificates differ by a multiplicative
constant. These constants are irrelevant to the main result of this paper. So we introduce a notation to suppress them.

Let~$h$ be a mixed term in~$\cR$ with~$\xi_x$-certificate~$a$ and~$\eta_y$-certificate~$b$. Then~$a$ and~$b$ are compatible
with respect to~$(\xi_x, \eta_y)$ because of the commutativity of~$\xi_x$ and~$\eta_y$. Set
\[  \cH(a, b) = \{ c h | c \in k\}. \]
The set consists of zero and mixed terms in~$\cR$ whose respective certificates are~$a$ and~$b$.
Clearly,~$\cH(a,b)$ is a one-dimensional linear subspace over~$k$.
In the sequel, whenever the notation~$\cH(a,b)$ is used,~$a$ and~$b$ are assumed to be compatible rational functions,
and~$\cH(a, b) \subset \cR$.  In particular, for any rational function~$f \in k(x,y)$,
the set~$f \cH(a, b)$ is a subset of~$\cR$. Indeed, it is the one-dimensional linear subspace spanned by~$fh$.
Moreover, let~$h^\prime$ be another mixed term in~$\cR$, and
$a^\prime$ and~$b^\prime$~be the corresponding certificates.
We consider
\[ \cH(a, b) \cH(a^\prime, b^\prime) = \left\{ r r^\prime | r \in \cH(a, b), r^\prime \in \cH(a^\prime, b^\prime) \right\}, \]
which is equal to the one-dimensional linear subspace spanned by~$h h^\prime$. By the definition of certificates, we also have
\[   \xi_x(\cH(a,b))= a \cH(a,b) \quad \text{and} \quad  \eta_y (\cH(a, b)) = b \cH(a, b). \]
More rules for calculations with~$\cH(a, b)$ are given below.
\begin{lemma} \label{LM:hrule}
For a field~$\mfield$, we
let~$\cH(a,b)$ and~$\cH(a^\prime,  b^\prime)$  be given above, and $f$ be a nonzero rational function in~$k(x,y)$.
\begin{enumerate}
\item[(i)] If~$\xi_x = \delta_x$ and~$\eta_y \in \{\si_y, \tau_y\}$, then
\[  f \cH(a, b) = \cH \left(   a + \frac{\xi_x(f)}{f}, \,  b \frac{\eta_y(f)}{f} \right)
\quad \text{and} \quad
\cH(a, b) \cH(a^\prime, b^\prime)  = \cH(a + a^\prime, b b^\prime). \]
\item[(ii)] If both~$\xi_x$ and~$\eta_y$ are automorphisms, then
\[  f \cH(a, b) = \cH \left(   a \frac{\xi_x(f)}{f}, \,  b \frac{\eta_y(f)}{f} \right)
\quad \text{and} \quad  \cH(a, b) \cH(a^\prime, b^\prime)  = \cH(aa^\prime, b b^\prime). \]
\item[(iii)] If~$\xi_x \in \{\si_x, \tau_x\}$ and~$\eta_y= \delta_y,$ then
\[  f \cH(a, b) = \cH \left(   a \frac{\xi_x(f)}{f}, \,  b + \frac{\eta_y(f)}{f} \right)
\quad \text{and} \quad \cH(a, b) \cH(a^\prime, b^\prime)  = \cH(aa^\prime, b + b^\prime). \]
\end{enumerate}
\end{lemma}
\begin{proof}
Let~$h$ be a mixed term in~$\cH(a,b)$. It is straightforward to verify that
the $\xi_x$-certificate and~$\eta_y$-certificate of~$fh$ are~$a + \xi_x(f)/f$ and~$b \eta_y(f)/f$,
respectively, if~$\xi_x$ is~$\delta_x$ and~$\eta_y$ is an automorphism in~$\{\xi_x, \eta_y\}$.
Let~$h \in \cH(a,b)$ and~$h^\prime \in \cH(a^\prime, b^\prime)$ with~$h h^\prime \neq 0$.
Then the $\xi_x$-certificate of~$hh^\prime$ is~$a + a^\prime$, and its $\eta_y$-certificate
is~$bb^\prime$. It follows that
\[ \cH(a, b) \cH(a^\prime, b^\prime)  = \cH(a + a^\prime, b b^\prime). \]
The first assertion holds.  The other two assertions can be proved in a similar way.
\end{proof}

\begin{lemma}\label{LM:similar}
Let~$h_1$ and~$h_2$ be two mixed terms over~$\mfield$.
If~$h_1$ and~$h_2$ are similar, then
\begin{enumerate}
  \item[(i)] $h_1+h_2$ is either equal to zero or similar to~$h_1$;
  \item[(ii)] for any~$L\in k(x, y)\langle \pa_x, \pa_y\rangle$,
      $L(h_1)$ is either equal to zero or similar to~$h_1$.
\end{enumerate}
\end{lemma}
\begin{proof} Let~$r\in k(x, y)$ be equal to~$h_1/h_2$. Then the first assertion follows from
the equality~$h_1+h_2=(1+{1}/{r})h_1$. Since~$h_1$ is a mixed term,
its successive derivatives and ($q$-)shifts are all similar to~$h_1$. The second assertion holds.
\end{proof}
\begin{remark}
Let~$h, h_1, h_2$ be three mixed terms. If~$h=h_1+h_2$, then the three terms are similar.
This is because~$\xi_x(h)=\xi_x(h_1)+\xi_x(h_2)$ and~$\eta_y(h)=\eta_y(h_1)+\eta_y(h_2)$.
\end{remark}

The next example illustrates how a linear differential or recurrence operator applies to mixed terms.
\begin{example} \label{EX:apply}
Let us consider how to apply $D_x$ and~$S_x$ to~$r \cH(u, v)$ with~$r, u, v \in k(x, y)$.
First, $D_x(r \cH(u, v)) = \left(\delta_x(r) + r u\right) \cH(u, v).$
Putting~$L_1 = D_x + u$, we rewrite the above relation as~$D_x(r \cH(u, v)) = L_1(r) \cH(u, v)$.
An easy induction shows that
\[  D_x^i (r \cH(u,v)) = L_i(r) \cH(u, v), \]
where~$L_i \in k(x, y)\langle D_x \rangle$ whose coefficients have a common denominator
that divides some power of~$\den(u)$.  Moreover,
the denominator of~$L_i(r)$ divides~$(\den(u) \den(r))^{i+1}$.

Let~$M_i=\left(\prod_{j=0}^{i-1} \si_x^j(u) \right) S_x^i$ for~$i>0$.
Then
\[ S_x^i (r\cH(u, v)) = M_i(r) \cH(u,v). \]
So the denominator of~$M_i(r)$ divides the
product~$\den\left(\prod_{j=0}^{i-1} \si_x^j(u)\right) \cdot \den\left(\si_x^i(r)\right).$
\end{example}

\begin{center}
\begin{figure}
\begin{center}
\begin{tabular}{|c|l|c|}
\hline
        Case                 & $(\xi_x, \eta_y)$ &  Mixed term~$h$  \\    \hline
D-$\Delta$      & $(\de_x,\si_y)$        & $\exists\ f \in k(x,y), \alpha, \beta \in k(x), \gamma \in k(y), \,\,
h   \in f\cH\left(y \frac{\de_x(\beta)}{\beta} + \alpha, \beta  \gamma \right)$ \\
D-$\Delta_q$  & $(\de_x,\tau_y)$       & $ \exists\ f \in k(x,y), \alpha \in k(x), \beta \in k(y), \,\,
h \in  f \cH\left( \alpha, \beta \right)$  \\
$\Delta$-$\Delta_q$    & $ (\si_x,\tau_y)$       & $ \exists\ f \in k(x,y), \alpha \in k(x), \beta \in k(y), \,\,
h \in  f \cH\left( \alpha, \beta \right)$  \\ \hline
$\Delta$-D      & $(\si_x,\de_y)$        &   $\exists\ f \in k(x,y), \alpha, \beta \in k(y), \gamma \in k(x), \,\,
h \in f \cH \left(\beta  \gamma,   x \frac{\de_y(\beta)}{\beta} + \alpha \right) $ \\
$\Delta_q$-D  & $(\tau_x,\de_y)$       & $ \exists\ f \in k(x,y), \alpha \in k(y), \beta \in k(x), \,\,
h \in  f \cH(\beta, \alpha)$ \\
$\Delta_q$-$\Delta$    & $(\tau_x,\si_y)$       & $\exists\ f \in k(x,y), \alpha \in k(y), \beta \in k(x), \,\,h \in  f \cH(\beta, \alpha)$  \\
\hline
\end{tabular}
\end{center}
\caption{Decomposition of mixed terms} \label{FIG:decomp}
\end{figure}
\end{center}

According to the structure of compatible rational functions given in
Figure~\ref{FIG:structure}, we obtain a multiplicative decomposition
of mixed terms, which is described in Figure~\ref{FIG:decomp}.  In
each case of the figure, a mixed term~$h$ is re-expressed by a
\emph{structural decomposition\/} as~$h = f h^\prime$ (see
Definion~\ref{DEF:structure} below).  An interesting property of such
a decomposition is that the certificates of~$h^\prime$ are expressed
by \emph{univariate\/} functions~$\alpha, \beta$ and~$\gamma$, as
given in Figure~\ref{FIG:decomp}. This fact will be crucial to prove
our existence criterion for telescopers.

Let us verify the conclusion in the differential-difference case in Figure~\ref{FIG:decomp}.
Assume that~$a$ and~$b$ are the $\xi_x$-certificate and~$\eta_y$-certificate of~$h$, respectively.
According to
the conclusion given in the differential-difference case in Figure~\ref{FIG:structure},
there exist~$f\in k(x, y)$, $\alpha, \beta \in k(x)$, and~$\gamma \in k(y)$ such that
\[  a = \frac{\de_x(f)}{f} + y \frac{\de_x(\beta)}{\beta} + \alpha \quad \text{and} \quad
b = \frac{\si_y(f)}{f} \cdot \beta \cdot \gamma. \]
Thus, the conclusion in the first row of Figure~\ref{FIG:decomp} holds by the first assertion in Lemma~\ref{LM:hrule}.
The rest of conclusions in Figure~\ref{FIG:decomp} can be verified likewise.

\begin{remark}
In the case when both~$\xi_x$ and~$\eta_y$ are shift operators, the Ore--Sato theorem says that every hypergeometric term can be
decomposed into the product of a rational function and a factorial term, (see also Corollary 4 in~\citep{AbramovPetkovsek2002a}).
The criterion  for the existence of telescopers given by~\citet{Abramov2003} is based on the Ore--Sato theorem.
\end{remark}

\begin{define} \label{DEF:structure}
Let~$h$ be a mixed term over~$\mfield$, as given in Figure~\ref{FIG:decomp}.
We say that~$f h^\prime$ is a {\em structural decomposition\/} of~$h$ if~$f$ is a rational function given in the
row corresponding to~$(\xi_x, \eta_y)$ in Figure~\ref{FIG:decomp}, and $h^\prime$ is the mixed term such that~$h=f h^\prime$.
\end{define}

\subsection{Split polynomials} \label{SUBSECT:split}
In the discrete case, integer-linear polynomials are used
to describe the existence criterion for telescopers of hypergeometric terms.
Split polynomials play a similar role as integer-linear polynomials in the mixed cases.

\begin{define}\label{DEF:splitpoly}
A polynomial~$p\in k[x, y]$ is said to be~{\emph{split\/}} if it is of the form~$p_1(x)p_2(y)$
with~$p_1\in k[x]$ and~$p_2\in k[y]$.
More generally, a rational function~$r \in k(x, y)$ is said to be~{\emph{split\/}} if it is of the form~$r_1(x)r_2(y)$
with~$r_1\in k(x)$ and~$r_2\in k(y)$.
\end{define}
A rational function~$f \in k(x,y)$ can always be decomposed
as~$f_1(x)f_2(y)f_3(x, y)$, where~$f_1 \in k(x)$, $f_2 \in k(y)$ and
neither~$\num(f_3)$ nor~$\den(f_3)$ has split factors except
constants. We call~$f_1 f_2$ and~$f_3$ the \emph{split\/} and
\emph{non-split parts\/} of~$f$, respectively. Both are defined up to a
nonzero multiplicative constant.

\begin{remark} \label{RE:split}
For a polynomial~$p\in k[x, y]$, one may decide whether it is split by
comparing all monic normalized coefficients of~$p$ with
respect to~$y$: $p$~is split if and only if all those are equal.  In
an implementation, one would abort early as soon as a mismatch is
found.
\end{remark}

The next definition is fundamental for our existence criterion.
\begin{define} \label{DEF:properness}
Let~$h$ be a mixed term over~$\mfield$. Assume that~$f h^\prime$ is a structural decomposition of~$h$.
We say that~$h$ is {\em proper\/} if~$\den(f)$ is split.
\end{define}

The lemma below shows that the properness of a mixed term is independent of the choice of its structural decompositions.
\begin{lemma} \label{LM:den}
Let~$h=f_1h_1=f_2h_2$ be a mixed term over~$\mfield$, where~$f_1, f_2$ are in~$k(x,y)$.
Assume that~$f_1h_1$ is a structural decomposition of~$h$. Moreover, assume that~$h_1 \in \cH(a_1, b_1)$
and~$h_2 \in \cH(a_2, b_2)$.
Then
\[  \mbox{$\den(f_1)$ is split \quad \text{if and only if} \quad $\den(f_2)$ is split,} \]
provided that one of the conditions in the following table holds:

\medskip
\begin{center}
\begin{tabular}{|c|c|}  \hline
$(\xi_x, \eta_y) =(\delta_x, \si_y)$ & either~$\den(a_2)$ or~$b_2$ is split;\\ \hline
$(\xi_x, \eta_y) =(\delta_x, \tau_y)$ & either~$\den(a_2)$ or~$b_2$ is split;\\ \hline
$(\xi_x, \eta_y) =(\si_x, \delta_y)$ & either~$a_2$ or~$b_2$ is split; \\ \hline
$(\xi_x, \eta_y) = (\si_x, \tau_y)$ & either~$a_2$ or~$\den(b_2)$ is split;\\ \hline
$(\xi_x, \eta_y) =(\tau_x, \delta_y)$ & either~$a_2$ or~$\den(b_2)$ is split;\\ \hline
$(\xi_x, \eta_y) = (\tau_x, \si_y)$ & either~$a_2$ or~$b_2$ is split. \\ \hline
\end{tabular}
\end{center}
\medskip

\end{lemma}
\begin{proof}
Assume that~$(\xi_x, \eta_y) = (\de_x, \si_y)$. Then, by the conclusion in the first row of Figure~\ref{FIG:decomp},
there exist~$\alpha, \beta \in k(x)$ and~$\gamma\in k(y)$ such that
\[a_1 = y\frac{\de_x(\beta)}{\beta} + \alpha \quad \text{and} \quad b_1 = \beta \gamma.\]
To show the equivalence in the first assertion, it suffices to prove that~$f_1/f_2$ is split.

\medskip \noindent
{\em Case 1.} Assume that~$\den(a_2)$ is split. Since~$f_1h_1=f_2h_2$, we have
\[\frac{\de_x(f_1)}{f_1} + y\frac{\de_x(\beta)}{\beta} + \alpha  = \frac{\de_x(f_2)}{f_2} + a_2,\]
which implies that
\begin{equation}\label{EQ:logd}
\frac{\de_x(f_1/f_2)}{f_1/f_2} = a_2 -y\frac{\de_x(\beta)}{\beta} - \alpha
\end{equation}
The denominator of the right-hand side of~\eqref{EQ:logd}
is split, since~$\den(a_2)$ is split and~$\alpha, \beta$ are in~$k(x)$.
Suppose that~$p\in k[x, y]$ is a nontrivial and non-split irreducible factor of~$\num(f_1/f_2)\cdot \den(f_1/f_2)$.
Then~$p$ divides the denominator of the logarithmic derivative of~$f_1/f_2$, a contradiction to~\eqref{EQ:logd}. Thus, $f_1/f_2$ is split.

\medskip \noindent
{\em Case 2.}
Assume that~$b_2$ is split. Since~$f_1h_1=f_2h_2$, we have
\[\frac{\si_y(f_1)}{f_1} \beta \gamma  = \frac{\si_y(f_2)}{f_2}  b_2,\]
which implies that
\begin{equation}\label{EQ:shiftq}
\frac{\si_y(f_1/f_2)}{f_1/f_2} = \frac{b_2}{\beta \gamma}.
\end{equation}
Note that the right-hand side of~\eqref{EQ:shiftq} is split. Suppose that~$p\in k[x, y]$ is a nontrivial and non-split irreducible factor
of the polynomial~$\num(f_1/f_2)\cdot \den(f_1/f_2)$. Then there exists~$\ell \in \bN$ such that
\[\si_y^\ell(p)\mid  \num(f_1/f_2)\cdot \den(f_1/f_2) \quad \text{but} \quad \si_y^{\ell+1}(p)\nmid \num(f_1/f_2)\cdot \den(f_1/f_2). \]
This implies that~$\si_y^{\ell+1}(p)$ is a factor of either the numerator or the denominator of the rational function~${\si_y(f_1/f_2)}/({f_1/f_2})$,
a contradiction to~\eqref{EQ:shiftq}.
Again, $f_1/f_2$~is split.

A similar argument proves that the assertion holds in the cases that~$(\xi_x, \eta_y)=(\delta_x, \tau_y)$
and~$(\xi_x, \eta_y)=(\si_x, \tau_y)$. The assertion for other cases can be proved by interchanging
the roles of~$x$ and~$y$.
\end{proof}
\section{Telescopers for mixed hypergeometric terms} \label{SECT:telescoper}
The method of creative telescoping was first formulated and popularized in a series of papers by
Zeilberger and his collaborators in the early 1990's~\citep{Almkvist1990,  Zeilberger1990c, Zeilberger1990, Zeilberger1991, Wilf1992}.
To illustrate the idea of this method, we consider the problem of finding a linear recurrence equation
for the integral (if there exists one):
\[H(x) := \int_{0}^{+\infty} h(x, y) \, dy,\]
where~$h(x, y)$ is a mixed term over~$\left(k(x, y), \, (\si_x, \de_y)\right)$.
Suppose that all integrals occurring in the derivation below are well-defined.
The key step of creative telescoping tries to find a nonzero
linear recurrence operator $L(x, S_x)$ in $k(x)\langle S_x \rangle$ such that
\begin{equation}\label{EQ:teleshift}
    L(x, S_x)(h) = D_y(g),
\end{equation}
for some mixed term~$g$ over~$k(x, y)$.

Applying the integral sign to
both sides of~\eqref{EQ:teleshift} yields
\[L(x, S_x)(H(x)) = g(x, +\infty)-g(x, 0).\]
This further implies that~$L(x, S_x)$ is indeed the recurrence relation satisfied by~$H(x)$
under certain nice boundary condition, say~$g(x, +\infty)=g(x, 0)$.
For example, consider the integral
\[A(x)=\int_{0}^{+\infty} y^{x-1}\exp(-y) \, dy. \]
The differential variant \cite{Almkvist1990} of Zeilberger's algorithm
brings us a pair~$(L, g)$ with
\[L = S_x-x\quad \text{and}\quad g=-y^x\exp(-y).\]
Note that~$g(x, +\infty)=g(x, 0)=0$, which implies that~$L(A(x))=A(x+1)-xA(x)=0$.
So we recognize that $A(x)=\Gamma(x)$ since the initial value~$A(1)$ is equal to~$1$.
For more interesting examples, see the appendix of~\citep{Almkvist1990}
or Koepf's book~\citep[Chapters~10--13]{Koepf1998}.

\begin{define}\label{DEF:telescoper}
Let~$h$ be a mixed term over~$\mfield$.
A nonzero linear operator~$L(x, \pa_x)\in k(x)\langle \pa_x \rangle$ is called a \emph{telescoper of type\/~$(\pa_x, \pa_y)$} for~$h$
if there exists another mixed term $g$ such that
\begin{equation}\label{EQ:telescoper}
L(x, \pa_x)(h) = \Delta_y(g).
\end{equation}
\end{define}

For a given mixed term, when does a telescoper of certain type exist?
And how one can construct telescopers?
These are two basic problems related to the method of creative telescoping.
In the subsequent sections, we will answer
the first one for the mixed cases. More precisely, we solve the following problem, which is equivalent to the termination problem
of creative-telescoping algorithms for mixed inputs.

\medskip \noindent
{\bf Existence Problem for Telescopers.}\,\,
For a mixed term~$h$ over~$\mfield$,
find a necessary and  sufficient condition on the existence
of telescopers of type~$(\pa_x, \pa_y)$ for~$h$.

\section{Exact terms and additive decompositions} \label{SECT:ad}
For a univariate hypergeometric term~$H(y)$, the Gosper algorithm~\citep{Gosper1978} decides whether it is hypergeometric summable \wrt~$y$,
i.e., whether $H=(S_y - 1)(G)$ for some hypergeometric term~$G$.
Based on the Gosper algorithm,
\cite{Zeilberger1990c, Zeilberger1990} developed his fast version of creative-telescoping algorithms
for bivariate hypergeometric terms. \cite{Almkvist1990} presented
a continuous analogue of the Gosper algorithm for deciding the hyperexponential integrability,
which leads to a fast algorithm for hyperexponential telescoping. From the viewpoint of creative telescoping,
the Gosper algorithm and its continuous analogue decide whether
the identity operator,~$1$, is a telescoper for the inputs.

The following notion of \emph{exact terms\/} is motivated in the
differential case by the existence of an underlying exact form.  This
differential-form point of view was used in a recent work of one of
the authors \citep{ChenKauersSinger2012}.

\begin{define} \label{DEF:exact}
Let~$h$ be a mixed term over~$\mfield$. We say that~$h$
is~\emph{exact~\wrt\/~$\pa_y$} if there exists a mixed term~$g$ such that~$h=\Delta_y(g)$,
where~$\Delta_y$ is defined in~\eqref{DEF:delta}.
\end{define}
\begin{remark}
In~\citep{AbramovPetkovsek2002b, GeddesLeLi2004}, an exact term is traditionally called~\emph{($q$-)hypergeometric summable\/} term in
the discrete case, and a~\emph{hyperexponential integrable\/} function in the continuous case, respectively.
For each choice of~$\pa_x$ in~$\{D_x, S_x, T_x\}$,
it is clear that every exact term~\wrt~$\pa_y$
has a telescoper of type~$(\pa_x, \pa_y)$:
for instance~$1$ is such a telescoper.
\end{remark}

The notion of spread polynomials is defined by~\citet{Abramov2003} for
establishing his criterion on the existence of telescopers for hypergeometric terms.
The following definition translates it into the continuous and $q$-discrete cases.
\begin{define} \label{DEF:spread}
Let~$K$ be a field of characteristic zero and~$\de_z, \si_z, \tau_z$ be the usual derivation,
shift and $q$-shift operators over~$K[z]$, respectively.
For a polynomial~$a\in K[z]$, we say that:
\begin{enumerate}
\item[(i)] $a$ is {\em $\de_z$-spread\/} if every nontrivial irreducible factor of~$a$
has multiplicity~$>1$;
\item[(ii)] $a$ is {\em $\si_z$-spread\/} if, for every nontrivial irreducible factor~$b$ of~$a$,
$\si_z^i(b) \mid a$ for some nonzero integer~$i$;
\item[(iii)] $a$ is {\em $\tau_z$-spread\/} if, for every nontrivial irreducible factor~$b$ of~$a$ with~$z \nmid b$,
$\tau_z^i(b) \mid a$ for some nonzero integer~$i$.
\end{enumerate}
\end{define}

(Note that Case~(iii) of this definition makes no constraint on the
multiplicity of~$z$ in~$b$.  This is because we shall only consider
$\tau_z$-spread \emph{non-split\/} polynomials in what follows.)

The following proposition is a mixed analogue of Theorem~8 in~\citep{Abramov2003},
which relates exact terms to spread polynomials.

\begin{proposition} \label{PROP:sum}
Let~$h=fh^\prime$ be a mixed term over~$\mfield$, where~$f\in k(x, y)$ and~$h^\prime\in \cH(u, v)$.
Then the following statements hold:
\begin{enumerate}
\item[(i)]  Case~$\eta_y=\delta_y$: If~$\den(v)$ is split and~$h$ is exact~\wrt~$D_y$, then the non-split
part of~$\den(f)$ is $\de_y$-spread.
\item[(ii)] Case~$\eta_y = \si_y$: If~$v$ is split and~$h$ is exact~\wrt~$S_y$, then the non-split
part of~$\den(f)$ is $\si_y$-spread.
\item[(iii)] Case~$\eta_y = \tau_y$: If~$v$ is split and~$h$ is exact~\wrt~$T_y$, then the non-split
part of~$\den(f)$ is $\tau_y$-spread.
\end{enumerate}
\end{proposition}
\begin{proof}
Let~$p$ be a non-split irreducible factor of~$\den(f)$ with $\deg_x p>0$ and~$\deg_y p >0$.

To prove the first assertion, we assume that~$\den(v)$ is split and~$h=D_y(g)$, where~$g$ is either equal to zero or similar to~$h$.
Then~$g=rh^\prime$ for some~$r \in k(x, y)$.
It follows from~$h=D_y(g)$ that~$f h^\prime = D_y(r h^\prime)$, which implies that
\[f = \de_y(r) + rv. \]
Since~$\den(v)$ is split, $p$ is an irreducible factor of~$\den(r)$.
So there exists an integer~$i>1$ such that~$p^i\mid \den(\de_y(r))$ and~$p^{i}\nmid \den(r)$.
Therefore, $p^i \nmid \den(rv)$ and $p^i \mid \den(f)$.
The first assertion holds.

To prove the second assertion, we assume that~$v$ is split and~$h=(S_y-1)(g)$, where~$g$ is either zero or similar to~$h$.
Then~$g=rh^\prime$ for some~$r\in k(x, y)$. From~$h=(S_y-1)(g)$, it follows that~$f h^\prime = (S_y-1)(r h^\prime)$.
Hence, we get
\[f = v\si_y(r) - r.\]
Since~$v$ is split, $p \mid \den(r)$ or~$p \mid \den(\si_y(r))$.
So, the set
\[L := \{ \ell\in\bZ \text{ such that } \si_y^\ell(p) \mid \den(r) \}\]
is finite and nonempty (consider $0\in L$ and $-1\in L$, respectively).
Therefore, there exist~$i, j \in \bZ$ with~$i>j$ such that
$i-1 \in L$, $i \not\in L$, $j-1 \not\in L$, $j \in L$, that is
\[\si_y^i(p) \mid \den(\si_y(r)), \,\, \si_y^i(p) \nmid \den(r), \,\,
\si_y^j(p) \nmid \den(\si_y(r)), \,\,
\text{and} \,\, \si_y^j(p) \mid \den(r) .\]
It follows that from the above equation both~$\si_y^i(p)$
and~$\si_y^j(p)$ divide~$\den(f)$.  If $i$~is nonzero, then both $p$
and~$\sigma^i(p)$ divide~$\den(f)$, so the non-split part of~$\den(f)$
is $\sigma$-spread.  Otherwise, $j$~must be nonzero.  Then the
non-split part is again $\sigma$-spread since both $\sigma^j(p)$
and~$p$ divide~$\den(f)$.  In any case, the second assertion holds.

The third assertion can be proved in the same vein as in the second case. We only need to note that an
irreducible factor~$p$ of the non-split part of~$\den(f)$ has at least two terms, which
implies that~$\tau_y^i(p)$ and~$\tau_y^j(p)$ are coprime if~$i \neq j$.
\end{proof}

The next notion to be introduced, related to exact terms, is that of
additive decompositions.

An algorithm by~\citet{Abramov2001, AbramovPetkovsek2002b} decomposes
a hypergeometric term~$H(y)$ into the sum~$\Delta_y(H_1) + H_2$, where~$H_2$ is minimal in some sense.
Such a decomposition is called an \emph{additive decomposition\/} for~$H$ \wrt~$y$.
Abramov and Petkov{\v s}ek's algorithm generalizes the capability of the Gosper algorithm in the sense that~$H$ is
hypergeometric summable if and only if~$H_2$ is zero.
In the continuous case, an algorithm to decompose a hyperexponential
function $H(y)$ as $D_y(H_1)+H_2$, where $H_1$ and~$H_2$ are either
zero or hyperexponential, is part of the proof of~Lemma~4.2
in~\cite{Davenport1986}.  This remained unknown to
\cite{GeddesLeLi2004}, who later described a similar additive
decomposition as a continuous analogue of Abramov and Petkovsek's
algorithm, but also proved that~$H_2$~satisfies certain minimality
requirement.
On the other hand, a $q$-discrete analogue is presented
in~\citep{ChenHouMu2005}.
When~$H$ is a rational function, additive decompositions
are more classical;
they were presented by ~\citet{Ostrogradsky1845} and~\citet{Hermite1872} for the
continuous case, and by~\citet{Abramov1975, Abramov1995b} for the discrete and $q$-discrete cases.

For a mixed term~$h$ over~$\mfield$, we can perform three kinds of
additive decompositions~\wrt~$y$ according to the choice of~$\eta_y$.
We recall now the notions related to additive decompositions for later
use.  Additive decompositions will be defined in
Definition~\ref{DEF:additive} below, after three steps of preliminary
material.

First, we borrow from~\citep{Abramov2001, AbramovPetkovsek2002b},
\citep{GeddesLeLi2004}, and~\citep{ChenHouMu2005} different notions of
reduced rational functions.

\begin{define} \label{DEF:reduced}
Let~$K$ be a field of characteristic zero and~$f$ be a rational function in~$K(z)$.
Denote by~$\de_z, \si_z,$ and~$\tau_z$ the usual derivation, shift and $q$-shift operators with
respect to~$z$ on~$K(z)$, respectively, and set~$a=\num(f)$, $b = \den(f)$. We say that
\begin{enumerate}
\item[(i)] $f$ is {\em $\de_z$-reduced\/} if~$\gcd(b, a-i\de_z(b))=1$ for all~$i \in \bZ$;
\item[(ii)] $f$ is {\em $\si_z$-reduced\/} if~$\gcd(b, \si_z^i(a))=1$ for all~$i \in \bZ$; and
\item[(iii)] $f$ is {\em $\tau_z$-reduced\/} if~$\gcd(b, \tau_z^i(a))=1$ for all~$i \in \bZ$.
\end{enumerate}
\end{define}

The next lemma reveals a connection between reduced rational functions and split ones. Its proof is based on the structure
of compatible rational functions given in Figure~\ref{FIG:structure}.

\begin{lemma} \label{LM:reduced}
Let~$u, v\in k(x, y)$ be two compatible rational functions~\wrt~$(\xi_x, \eta_y) \in \Theta$.
\begin{enumerate}
\item[(1)] \emph{Case~$\eta_y=\de_y$ and $\de_y$-reduced~$v$.}
  If~$\xi_x\in \{\si_x, \tau_x\}$, then both $u$ and~$\den(v)$ are
  split.
\item[(2)] \emph{Case~$\eta_y=\si_y$ and $\si_y$-reduced~$v$.}  Then:
\begin{enumerate}
\item[(2a)] if~$\xi_x=\de_x$, then both~$\den(u)$ and~$v$ are split;
\item[(2b)] if~$\xi_x=\tau_x$,  then both~$u$ and~$v$ are split.
\end{enumerate}
\item[(3)] \emph{Case~$\eta_y=\tau_y$ and $\tau_y$-reduced~$v$.}
  Then:
\begin{enumerate}
\item[(3a)] if~$\xi_x=\de_x$, then both~$\den(u)$ and~$v$ are split;
\item[(3b)] if~$\xi_x=\si_x$, then both~$u$ and~$v$ are split.
\end{enumerate}
\end{enumerate}
\end{lemma}
\begin{proof}
Assume that~$\xi_x=\si_x$ and~$\eta_y = \de_y$.
By the conclusion in the fourth row of Figure~\ref{FIG:structure}, there exist~$f\in k(x, y)$, $\gamma\in k(x)$,
$\alpha, \beta\in k(y)$ such that
\begin{equation}\label{EQ:sd}
u = \frac{\si_x(f)}{f}\beta \gamma \quad \text{and} \quad v=\frac{\de_y(f)}{f} + x \frac{\de_y(\beta)}{\beta} + \alpha.
\end{equation}
Since~$v$ is~$\de_y$-reduced, every irreducible factor of~$\den(f) \cdot \num(f)$ is either in~$k[x]$ or in~$k[y]$
by Lemma~2 in~\citep{GeddesLeLi2004} and the second equality in~\eqref{EQ:sd}. Hence,~$\den(v)$ is split.
As $\den(f) \cdot \num(f)$ is split, it also follows from the first equality
that~$u$ is split. The first assertion holds when~$\xi_x=\si_x$.
Using the conclusion in the fifth row of Figure~\ref{FIG:structure} and a similar argument, one can show that the assertion
holds when~$\xi_x = \tau_x$.

To prove the second assertion, we first assume that~$\xi_x=\de_x$ and~$\eta_y = \si_y$.
By the conclusion in the first row of Figure~\ref{FIG:structure}, there exist~$f\in k(x, y)$,
$\alpha, \beta\in k(x)$, and~$\gamma\in k(y)$ such that
\begin{equation}\label{EQ:ds}
u = \frac{\de_x(f)}{f} + y \frac{\de_x(\beta)}{\beta} + \alpha  \quad \text{and} \quad v= \frac{\si_y(f)}{f}\beta \gamma.
\end{equation}
Let~$p$ be a non-split nontrivial irreducible factor of~$ \den(f) \cdot \num(f)$.
Then there exists a nonnegative integer~$i$ such that
\[\si_y^i(p)\mid \den(f)\cdot \num(f) \quad \text{and} \quad \si_y^{i+\ell}(p)\nmid  \den(f) \cdot \num(f)
\,\, \mbox{for all~$\ell>0$}, \]
which, together with the second equality in~\eqref{EQ:ds},~$\beta \in k(x)$ and~$\gamma \in k(y)$, implies that
\begin{enumerate}
\item[$\bullet$]  $\si_y^{i+1}(p)\mid \num(v)$ and~$\si_y^{j}(p)\mid \den(v)$ if
$$\si_y^i(p) \mid \num(f), \quad \si_y^j(p) \mid \num(f), \quad \text{but} \quad \si_y^{j-1}(p) \nmid \num(f);$$
\item[$\bullet$]  $\si_y^{i+1}(p)\mid \den(v)$ and~$\si_y^{j}(p)\mid \num(v)$ if
$$\si_y^i(p) \mid \den(f), \quad \si_y^j(p) \mid \den(f), \quad \text{but} \quad \si_y^{j-1}(p) \nmid \den(f).$$
\end{enumerate}
Either case leads to a contradiction to the assumption that~$v$ is~$\si_y$-reduced.
Hence,~$f$ is split and so is~$v$. By the first equality in~\eqref{EQ:ds},~$\den(u)$ is split. This proves Assertion~(2a).

Let~$\xi_x = \tau_x$. By the conclusion in the sixth row of Figure~\ref{FIG:structure}, there exist~$f \in k(x,y)$,
$\alpha \in k(y)$ and~$\beta \in k(x)$ such that
\begin{equation} \label{EQ:qs}
u = \frac{\tau_x(f)}{f} \beta \quad \text{and} \quad v = \frac{\si_y(f)}{f} \alpha.
\end{equation}
The same argument used in Case (2a) implies that~$f$ is split, and so is~$v$.
By the first equality in~\eqref{EQ:qs}, $u$ is also split.

Finally, we let~$\eta_y=\tau_y$. Then Assertions (3a) and (3b) hold by
the
conclusions in second and third rows of Figure~\ref{FIG:structure} and similar arguments used in the proofs
of Assertions~(2a) and~(2b), respectively.
\end{proof}

Second, we recall the notions of squarefree, shift-free and~$q$-shift-free polynomials.
\begin{define} \label{DEF:sfree}
Let~$K$ be a field of characteristic zero, and~$a$ be a nonzero polynomial in~$K[z]$.
Denote by~$\de_z, \si_z,$ and~$\tau_z$ the usual derivation, shift and $q$-shift operators with
respect to~$z$ on~$K(z)$, respectively.
\begin{enumerate}
  \item [(i)]  $a$ is said to be~\emph{squarefree\/} if~$\gcd\left(a, \de_z(a) \right)=1$.
  \item [(ii)] $a$ is said to be~\emph{shift-free\/} if~$\gcd(a, \si_z^i(a))=1$ for every nonzero integer~$i$.
  \item [(iii)] Let~$a=z^s \tilde{a}$ with~$\tilde{a}\in K[z]$ and~$z \nmid \tilde{a}$.
  Then~$a$ is said to be~\emph{$q$-shift-free\/}
  if
  $$\gcd \left( \tilde{a},  \tau_z^i(\tilde{a})\right)=1 \quad \mbox{for every nonzero integer~$i$}.$$
\end{enumerate}
\end{define}

Note that every non-split polynomial in~$k[x,y]$ has at least two terms.
For a polynomial in~$K[z]$ having at least two terms, it is not $\de_z$-spread (resp.\ $\si_z$-spread, $\tau_z$-spread)
if it is squarefree (resp.\ shift-free, $q$-shift-free) with respect to~$z$.  However, the converses are
false, as shown in the examples below.
\begin{example}
Let~$p=z^2(z+1)$. Since the multiplicity of~$z+1$ in~$p$ in~$1$, $p$ is not~$\de_z$-spread.
However, $p$ is not squarefree with respect to~$z$, because~$\gcd(p, \de_z(p))=z\notin K$.
\end{example}

\begin{example}
Let~$p=z(z+1)(z+1/2)$. Since any nontrivial shift of~$(z+1/2)$ does not divide~$p$, $p$ is not~$\si_z$-spread.
However, $p$ is not shift-free with respect to~$z$, because~$\gcd(p, \si_z(p))=z+1$.
\end{example}

\begin{example}
Let~$p=(z+1)(z+q)(z+2)$ and~$q\in K$ such that~$q^i\neq 2$ for any~$i\in \bZ$.
Since~$\tau_z^i(z+2)\nmid p$ for any nonzero~$i\in \bZ$, $p$ is not~$\tau_z$-spread.
However, $p$ is not $q$-shift-free with respect to~$z$, because~$\gcd(p, \tau_z(p))=z+1$.
\end{example}

Finally, we define the three additive decompositions in the setting of mixed terms.
\begin{define} \label{DEF:additive}
Let~$h$ be a mixed term over~$\mfield$. Assume that
\begin{equation} \label{EQ:add}
h = \Delta_y(h_1) + h_2
\end{equation}
where~$h_1$ is a mixed term, and~$h_2$ is equal to either zero or a mixed term
of the form
\begin{equation} \label{eq:additive-decomposition}
h_2 \in r \cH(u, v)
\end{equation}
for some~$r, u, v \in k(x, y)$ satisfying
\begin{enumerate}
\item[(a)] $\den(r)$ is squarefree with respect to~$y$, and~$v$ is $\delta_y$-reduced if~$\eta_y=\delta_y$;
\item[(b)] $\den(r)$ is shift-free with respect to~$y$, and~$v$ is~$\si_y$-reduced if~$\eta_y = \si_y$;
\item[(c)] $\den(r)$ is $q$-shift-free with respect to~$y$, and~$v$ is~$\tau_y$-reduced if~$\eta_y = \tau_y$.
\end{enumerate}
\noindent We call~\eqref{EQ:add} an {\em additive decomposition\/} of~$h$ with respect to~$\pa_y$.
\end{define}

Additive decompositions with respect to~$D_y$  can be computed by the algorithms
described in~\citep{Davenport1986} and~\citep{GeddesLeLi2004}.
Additive decompositions with respect to~$S_y$ and~$T_y$ can be computed by the algorithms in~\citep{Abramov2001}
and~\citep{ChenHouMu2005}, respectively. In particular, additive decompositions always exist, although they are not unique.
We remark that the additive decompositions given in Definition~\ref{DEF:additive} are weaker than those in~\citep{Abramov2001}
and~\citep{GeddesLeLi2004}. For example, $h_2$ is not necessarily equal to zero when~$h$ is an exact term.

\section{A criterion on the existence of telescopers for mixed terms} \label{SECT:criteria}
The Fundamental Theorem in~\citep{Wilf1992} states that there exist telescopers for
proper mixed terms. However, properness is just a sufficient condition.
In this section, we present a necessary and sufficient condition  on the existence
of telescopers for mixed terms.

\begin{define} \label{DEF:telescopable}
A mixed term~$h$ over~$\mfield$ is said to be~\emph{telescopable of type\/~$(\pa_x, \pa_y)$} if it has a
telescoper of type~$(\pa_x, \pa_y)$.
\end{define}

In this section, we show that
every proper mixed term is telescopable in~\S\ref{SUBSECT:suff},
and that every telescopable term is
a sum of an exact term and a proper one in~\S\ref{SUBSECT:nece}.

\subsection{Proper terms are telescopable} \label{SUBSECT:suff}
Wilf and Zeilberger present an elementary proof
of the existence of telescopers for proper hypergeometric terms~\cite[Theorem 3.1]{Wilf1992},
and indicate that their argument should be applicable to the mixed
setting. For the sake of completeness, we elaborate a proof
that every proper mixed term is telescopable.
Before presenting the proof, we need a few lemmas.

\begin{lemma}\label{LM:product}
Let~$h$ and~$H$ be two mixed terms over~$\mfield$ with~$\Delta_y(H)=0$, where~$\Delta_y$ is defined in~\eqref{DEF:delta}.
If~$h$ is telescopable of type~$(\pa_x, \pa_y)$, so is~$Hh$.
\end{lemma}
\begin{proof}
Let~$L \in k(x) \langle \pa_x \rangle$ be a telescoper of~$h$ such that~$L(h) = \Delta_y(g)$ for some mixed term~$g$.
Assume that~$H\in \cH(a, b)$. Since~$\Delta_y(H)=0$, $\eta_y(H)=\eta_y(1)H$.
Thus,~$b=\eta_y(1)$. It follows that~$b=0$ if~$\eta_y=\de_y$, and that~$b=1$
if either~$\eta_y = \si_y$ or~$\eta_y = \tau_y$.

By the compatibility conditions on~$a$ and~$b$ given in
Figure~\ref{FIG:compatible}, we derive that $a$~is nonzero and,
depending on~$\eta_y$, that one of the following formulas holds:
\emph{(i)\/} $\delta_y(a) = 0$; \emph{(ii)\/} $\sigma_y(a) = a$;
\emph{(iii)\/} $\tau_y(a) = a$.  At this point, we claim that~$a \in
k(x)$.  This is clear in case~(i).  In case (ii) or~(iii), write
$a=u/v$ for coprime $u$ and~$v$, and $v$~monic \wrt~$y$.  Then, either
$u/v = \si_y(u)/\si_y(v)$ or $u/v = \tau_y(u)/\tau_y(v)$, where all
fractions are reduced.  As a consequence, either $u = \alpha \si_y(u)$
and $v = \beta \si_y(v)$, or $u = \alpha \tau_y(u)$ and $v = \beta
\tau_y(v)$, in each case for $\alpha$ and~$\beta$ nonzero in~$k(x)$.
Case~(ii) leads to $\alpha = \beta = 1$, and $a \in k(x)$.  In
case~(iii), from $q$~not being a root of unity follows that $u$
and~$v$ are monomials in~$y$.  By coprimeness, at most one of them has
nonzero degree in~$y$.  If $u \in k(x)$, then $\alpha/\beta$~is equal
to~$q^j$ where~$j = \deg_y v$, so that, as $\alpha/\beta = 1$, $v$~is
in~$k(x)$ as well.  A symmetric argument leads to the fact that both
$u$ and~$v$ are in~$k(x)$, and so is~$a$.

We continue by distinguishing two cases according to~$\xi_x$:

\noindent
{\em Case 1.} Assume that~$\xi_x = \delta_x$. Then~$\pa_x(Hh)=\pa_x(H)h + H \pa_x(h)$,
so that~$(\pa_x - a) (Hh) = H \pa_x(h)$.
By an easy induction, this implies that
\begin{equation} \label{EQ:replace1}
   (\pa_x - a)^i(Hh) = H \pa_x^i(h) \quad \mbox{for all~$i \in \bN$}.
\end{equation}
Let~$L^\prime$ be the Ore polynomial obtained by replacing~$\pa_x^i$ in~$L$ with~$(\pa_x-a)^i$.
Then~$L^\prime$ belongs to~$k(x)\langle \pa_x \rangle$ because~$a$ does.
By~\eqref{EQ:replace1}, $L^\prime(Hh) = H L(h)$. Hence,
$$L^\prime(Hh) = H L(h) = H \Delta_y (g) = \Delta_y (H g),$$
in which the last equality follows from~$\Delta_y(H)=0$. The product~$Hh$ is telescopable.

\smallskip \noindent
{\em Case 2.} Assume that~$\xi_x = \si_x$ or~$\xi_x = \tau_x$. Then~$\pa_x(Hh)=\pa_x(H)\pa_x(h)$.
So
$$ a^{-1} \pa_x   (Hh) = H \pa_x(h),$$
 which, together with an easy induction, implies that
\[
    \left(\prod_{j=0}^{i-1} \si_x^j (a)\right)^{-1} \pa_x^i(Hh) = H \pa_x^i(h) \quad \mbox{for all~$i \in \bN$}.
\]
Let~$L^{\prime\prime}$ be the Ore polynomial obtained by replacing~$\pa_x^i$ in~$L$ with~$\left(\prod_{j=0}^{i-1} \si_x^j (a)\right)^{-1}
\pa_x^i$.
Then~$L^{\prime\prime}$ belongs to~$k(x)\langle \pa_x \rangle$ because~$a$ does.
Consequently,~$L^{\prime\prime}$ is a telescoper of~$Hh$ of type~$(\pa_x, \pa_y)$ by the same
argument as the one used in Case~1.
\end{proof}

\begin{lemma}\label{LM:closure}
Let~$h_1$ and~$h_2$ be two similar mixed terms over~$\mfield$.
Let~$\bV_y$ be the $k(x)$-vector space spanned
by
\[\left\{ \Delta_y^i (h_j) \, \mid \, i\in \bN, j \in \{1, 2\}\right\}, \]
where~$\Delta_y = \Delta_y$.
If both~$h_1$ and~$h_2$ have telescopers of type~$(\pa_x, \pa_y)$,
so does every nonzero element in~$\bV_y$.
\end{lemma}
\begin{proof} Assume that, for~$j=1,2$, $L_j(h_j)=\Delta_y(g_j)$
for some~$L_j \in k(x)\langle \pa_x \rangle$  and for some mixed term~$g_j$. Then
\[L_j \left(\Delta_y^i(h_j)\right)= \Delta_y \left( \Delta_y^i(g_j) \right) .\]
Hence, every nonzero element of $\left\{ \Delta_y^i(h_j) \, \mid \, i\in \bN, j \in \{1, 2\}\right\}$
is telescopable of~$(\pa_x, \pa_y)$.
By Lemma~\ref{LM:product}, a mixed term of the form~$f \Delta_y^i(h_j)$ for every~$f \in k(x)^*$
is also telescopable of the same type. Assume that~$H_1$ and~$H_2$ are two telescopable mixed terms.
Then a common left multiple of their telescopers
is a telescoper of the sum~$H_1+ H_2$. Hence, every nonzero element of~$\bV_y$ is telescopable.
\end{proof}

The \lq\lq noncommutative trick\rq\rq\ below was first used
in~\citep[Theorem~3.2]{Wegschaider1997} to transform a nonzero linear
recurrence operator not involving~$y$ to a telescoper.  Here, we
generalize it to transform an operator in~$k(x)\langle \pa_x, \pa_y
\rangle$ to a telescoper of type~$(\pa_x, \pa_y)$.  The result,
Lemma~\ref{LM:annihilator}, bases on formulas depending on the
operator types, which we state first.  However, only the cases when
$\eta_y$~is $\delta_y$ or~$\si_y$ are necessary in the proof of
Lemma~\ref{LM:sproper}, so Lemmas \ref{LM:noncomm-trick}
and~\ref{LM:annihilator} do not consider $q$-discrete operators.

\begin{lemma} \label{LM:noncomm-trick}
The following statements hold.
\begin{enumerate}
\item \emph{Case~$\eta_y = \si_y$.}  For any~$m \in \bN$, there
  exists~$Q_m\in k[y]\langle S_y \rangle$ satisfying
\[ \frac{y^m}{(-1)^m m!} (S_y-1)^m = 1 + (S_y-1) Q_m . \]
\item \emph{Case~$\eta_y = \delta_y$.}  For any $m \in \bN$, there
  exists~$Q_m\in k[y]\langle D_y \rangle$ satisfying
\[ \frac{y^m}{(-1)^m m!} D_y^m = 1 + D_y Q_m . \]
\end{enumerate}
\end{lemma}
\begin{proof}
The first assertion is just Theorem~3.2 in~\citep{Wegschaider1997}.
The second assertion follows from a straightforward continuous analogue of Wegschaider's proof
for the first assertion.
\end{proof}

\begin{lemma} \label{LM:annihilator}
Let~$h$ be a mixed term over~$\mfield$, under the restriction $\eta_y
\in \{ \delta_y, \si_y \}$.
If there exists a nonzero operator~$A\in k(x)\langle \pa_x, \pa_y \rangle$ such that~$A(h)=0$,
then~$h$ has a telescoper of type~$(\pa_x, \pa_y)$.
\end{lemma}
\begin{proof}
Assume that~$A(h)=0$ for some nonzero operator~$A$ in~$k(x)\langle \pa_x, \pa_y \rangle$.
By a repeated use of left-hand division, we can write
\begin{equation} \label{EQ:lrem}
A =\Delta_y^m(L(x, \pa_x) + \Delta_y M),
\end{equation}
where~$m$ is in~$\bN$, $\Delta_y$ is defined in~\eqref{DEF:delta}, $L$ is a nonzero operator in~$k(x)\langle \pa_x \rangle$ and~$M$  is
in~$k(x)\langle \pa_x, \pa_y\rangle$.
By Lemma~\ref{LM:noncomm-trick}, we have
\[
\frac{y^m}{(-1)^m m!}\Delta_y^m = \Delta_y Q + 1
\]
for some~$Q \in k[y]\langle \pa_y \rangle$, which, together with~\eqref{EQ:lrem}, implies that
\[  \frac{y^m}{(-1)^m m!} A = L + \Delta_y N \quad \mbox{for some~$N \in k(x, y)\langle \pa_x, \pa_y \rangle$.}\]
It follows from~$A(h)=0$ that~$L(h)=\Delta_y( - N(h))$. Hence,~$L$ is a telescoper for~$h$ of type~$(\pa_x, \pa_y)$.
\end{proof}

The next lemma paves the way to apply the multiplicative
decompositions in Figure~\ref{FIG:decomp}.  Its proof is reminiscent
of the linear algebra argument given by~\citet{Lipshitz1988}; however,
it bases on linear algebra and filtrations over~$k(x)$ instead of~$k$.
\begin{lemma}\label{LM:sproper}
Let~$h$ be a mixed term over~$\mfield$.
\begin{enumerate}
  \item [(i)] If~$(\xi_x, \eta_y)=(\delta_x, \si_y)$ and~$h\in \cH\left(y\frac{\de_x(\beta)}{\beta}, \beta\gamma\right)$ with~$\beta\in k(x)$
  and~$\gamma\in k(y)$, then~$h$ has a telescoper of type~$(D_x, S_y)$.
  \item [(ii)]  If~$(\xi_x, \eta_y)=(\si_x, \de_y)$ and~$h\in
  \cH\left(\beta,  x\frac{\de_y(\beta)}{\beta} + \alpha \right)$ with~$\alpha, \beta\in k(y)$, then~$h$ has a telescoper of type~$(S_x, D_y)$.
\end{enumerate}
\end{lemma}
\begin{proof}
To show the first assertion, we let
\[ a=\num(\beta), \,\, b=\den(\beta), \,\, s=\num(\gamma), \,\, t=\den(\gamma), \]
and
\[ u=\num\left(\frac{\de_x(\beta)}{\beta}\right), \,\,
v=\den\left(\frac{\de_x(\beta)}{\beta} \right).\]
A straightforward calculation yields that the $\de_x$-certificate and~$\si_y$-certificate of~$h$ are,
\[  y\frac{\de_x(\beta)}{\beta} = \frac{yu}{v} \quad \text{and} \quad  \beta\gamma = \frac{as}{bt}, \]
respectively. Note that~$a, b, u, v \in k[x]$ and~$s, t \in k[y]$.

Let~$\mathcal {F}_{N}$ be the linear subspace spanned by~$\{D_x^iS_y^j\mid i+j\leq N \}$
over~$k(x)$. Let $\mu$ be the maximum of the degrees in~$y$ of~$s$ and~$t$, and let
\[   \cW_N = \spa_{k(x)} \left\{ \frac{y^\ell h}{t(y+N-1)\cdots t(y)} \,\mid \,\ell \leq (\mu+1) N \right\}. \]
An easy induction on~$i$ and~$j$ yields
\[ D_x^iS_y^j(h)=\frac{w(x, y) h}{t(y+j-1)\cdots t(y)},~\mbox{where~$w \in k(x)[y]$ and~$\deg_y(w)\leq j\mu+i.$} \]
Hence,~$D_x^i S_y^j(h)$ belongs to~$\cW_N$ if~$i+j \leq N$.
Accordingly, there is a well-defined $k(x)$-linear map~$\phi_N$ from~$\cF_N$ to~$\cW_N$ that sends~$L$ to~$L(h)$
for all~$L \in \cF_N$.
Since the dimension of~$\mathcal {F}_N$ over~$k(x)$ is~$\binom{N+2}{2}$, while that of~$\mathcal{W}_{N}$ is~$(\mu+1) N+1$,
the kernel of~$\phi_N$ is nontrivial when~$N$ is sufficiently large.
Let~$A$ be a nonzero element of~$\ker(\phi_N)$. Then~$A(h)=0$. The first assertion follows
from Lemma~\ref{LM:annihilator}.

To show the second assertion, we let~$a=\num(\beta)$, $b=\den(\beta)$,
\[u=\num\left(x\frac{\de_y(\beta)}{\beta} + \alpha\right) \quad \text{and} \quad
v=\den\left(y\frac{\delta_y(\beta)}{\beta} + \alpha\right).\]
A straightforward calculation yields that the~$y$-certificate of~$h$ is
\[ x\frac{\de_y(\beta)}{\beta} + \alpha = \frac{u(x, y)}{v(y)}.\]
Note that~$a, b, v \in k[y]$, $u \in k[x, y]$.

Consider the linear
space~$\mathcal {F}_{N}$ spanned by~$\{S_x^iD_y^j\mid i+j\leq N \}$
over~$k(x)$. Let $\mu$ be the maximum of the degrees in~$y$ of
$a, b, u$ and~$v$, and let
\[   \cW_N = \spa_{k(x)} \left\{ \frac{y^\ell
h}{(bv)^N} \,\mid \, \ell\leq 2\mu N\right\}. \]
An easy induction on~$i$ and~$j$ yields
\[ S_x^iD_y^j(h)=\frac{w(x,
y)}{b^iv^j}h, \quad~\mbox{where~$w \in k(x)[y]$ with~$\deg_y(w)\leq (i+j)\mu$}. \]
Hence, $S_x^iD_y^j(h)$ belongs to~$\cW_N$ if~$i+j \leq N$.
Accordingly, there is well-defined a $k(x)$-linear map~$\psi_N$ from~$\cF_N$ to~$\cW_N$ that sends~$L$ to~$L(h)$
for all~$L \in \cF_N$.
Since the dimension
of~$\mathcal {F}_N$ over~$k(x)$ is~$\binom{N+2}{2}$, while that of~$\mathcal{W}_{N}$ is~$2\mu N+1$,
the kernel of~$\psi_N$ is nontrivial when~$N$ is sufficiently large.
Let~$A$ be a nonzero element in~$\ker(\psi_N)$. Then~$A(h)=0$.
The second assertion follows
from Lemma~\ref{LM:annihilator}.
\end{proof}

We are ready to show the main conclusion of this subsection.
\begin{theorem} \label{TH:fund}
Let~$h$ be a mixed term over~$\mfield$. If~$h$ is proper, then it is telescopable of type~$(\pa_x, \pa_y)$.
\end{theorem}
\begin{proof}
In the mixed setting, there are six cases to be considered.

First, we  assume that~$\xi_x=\de_x$ and~$\eta_y = \si_y$,
and prove that~$h$ has a telescoper of type~$(D_x, S_y)$.
By the conclusion in the first row of Figure~\ref{FIG:decomp} and Lemma~\ref{LM:den},
$h$~has a structural
decomposition
\[h\in \frac{A}{B\cdot C} \cH\left(y\frac{\de_x(\beta)}{\beta} + \alpha, \beta \gamma\right), \]
where~$A\in k[x, y], B\in k[x], C\in k[y]$, $\alpha, \beta\in k(x)$ and~$\gamma\in k(y)$.
By Lemma~\ref{LM:hrule}~(i),
\[  h = A U V, \quad \text{where } U \in \cH\left(\alpha + \frac{\de_x(1/B)}{1/B}, 1 \right),\ \text{and}\
V \in \cH\left(y\frac{\de_x(\beta)}{\beta}, \beta \frac{\si_y(1/C)}{1/C}\gamma\right). \]
Write~$A=\sum_{i=0}^m a_iy^i$ with~$a_i\in k[x]$.
Then we have $h = \sum_{i=0}^m G_iH_i$, where
\[G_i=a_iU\in \cH\left(\alpha + \frac{\de_x(a_i/B)}{a_i/B}, 1 \right)
\ \text{and}\
H_i=y^iV \in \cH\left(y\frac{\de_x(\beta)}{\beta}, \beta \frac{\si_y(y^i/C)}{y^i/C}\gamma\right).\]
Note that
\[ \alpha + \frac{\de_x(a_i/B)}{a_i/B}\in k(x) \quad \text{and} \quad
\frac{\si_y(y^i/C)}{y^i/C}\gamma\in k(y). \]
By Lemma~\ref{LM:sproper}~(i), the term $H_i$ has a telescoper of type~$(D_x, S_y)$ for all~$i$ with~$0\leq i \leq m$.
Since~$(S_y-1)(G_i)=0$, $G_iH_i$ has a telescoper of type~$(D_x, S_y)$
for all~$i$ with~$0\leq i \leq m$ by Lemma~\ref{LM:product}.
So~$h$ has a telescoper of the same type by Lemma~\ref{LM:closure}.
This completes the proof for the first case.

Second, we assume that~$\xi_x=\si_x$ and~$\eta_y = \de_y$, and show that~$h$ has a telescoper of type~$(S_x, D_y)$.
By the conclusion in the fourth row of Figure~\ref{FIG:decomp} and Lemma~\ref{LM:den},
$h$~has a structural decomposition
\[h \in \frac{A}{B\cdot C} \cH\left(\beta\gamma,  \alpha+x\frac{\de_y(\beta)}{\beta}\right), \]
where~$A\in k[x, y], B\in k[x], C\in k[y]$, $\alpha, \beta\in k(y)$ and~$\gamma\in k(x)$.

Write~$A=\sum_{i=0}^m a_iy^i$ with~$a_i\in k[x]$. A similar consideration as above leads to
\[h = \sum_{i=0}^m G_iH_i, \quad \text{where~$G_i\in \cH\left(\gamma \frac{\si_x(a_i/B)}{a_i/B}, 0\right)$
and~$H_i \in \cH\left(\beta, \alpha + x\frac{\de_y(\beta)}{\beta}\right)$}.\]
Note that~$\gamma \si_x(a_i/B)/(a_i/B) \in k(x)$ and~$\alpha\in k(y)$.
By Lemma~\ref{LM:sproper}~(ii), the term $H_i$ has a telescoper of type~$(S_x, D_y)$ for all~$i$ with~$0\leq i \leq m$.
Since~$D_y(G_i)=0$, $G_iH_i$ has a telescoper
of type~$(S_x, D_y)$ by Lemma~\ref{LM:product},
so does $h$ by Lemma~\ref{LM:closure}.

In the other four cases,~$(\xi_x, \eta_y)$ is equal to~$(\de_x, \tau_y)$, $(\si_x, \tau_y)$,
$(\tau_x, \delta_y)$ and~$(\tau_x, \si_y)$, respectively.
The presence of $q$-shift operators leads to simpler structural decompositions, which
enable us prove the existence of telescopers in a fairly straightforward way.

By Lemma~\ref{LM:den} and the conclusions in the second, third, fifth
and sixth rows of Figure~\ref{FIG:decomp},
$h$~has a structural decomposition
\[h \in \frac{A}{B\cdot C} \cH\left(\alpha, \beta\right), \]
where~$A \in k[x,y]$, $B \in k[x]$, $C \in k[y]$, $\alpha\in k(x)$ and~$\beta\in k(y)$.

Write~$A=\sum_{i=0}^m a_i y^i$ with~$a_i\in k[x]$.
It follows from Lemma~\ref{LM:hrule} that
\[  h = \sum_{i=0}^m G_i H_i, \quad
\mbox{where $G_i \in \frac{a_i}{B}\cH(\alpha, 1)$ and~$H_i \in \frac{y^i}{C}\cH(1, \beta).$}
\]
By Lemma~\ref{LM:closure}, it suffices to show that, for all~$i$ with~$0\leq i \leq m$, $G_i H_i$
has a telescoper of type~$(\pa_x, \pa_y)$.
Since~$\Delta_y(G_i)=0$, it suffices to show that~$H_i$ has a telescoper
by Lemma~\ref{LM:product}.
Since~$\Delta_x(H_i)=0$,~$\Delta_x$ is a telescoper of~$H_i$.
\end{proof}

\subsection{Characterization of telescopable terms} \label{SUBSECT:nece}
We have seen that exact terms and proper ones are all telescopable.
We now show that a telescopable term is the sum of an exact term and a proper one.
To this end, we present a few useful lemmas.
\begin{lemma} \label{LM:ns}
Let~$h, h_1$ and~$h_2$ be three mixed terms over~$\mfield$.
Assume that~$h = \Delta_y(h_1) + h_2$.
Then~$L$ is a telescoper of type~$(\pa_x, \pa_y)$ for~$h$
if and only if it is a telescoper for~$h_2$.
\end{lemma}
\begin{proof}
If~$h$ has a telescoper~$L$ in~$k(x)\langle \pa_x \rangle$, then~$L(h)= \Delta_y(g)$ for some
mixed term~$g$. Since~$L$ and~$\Delta_y$ commute with each other, $L(h_2)=\Delta_y(g-L(h_1))$.
So~$L$ is a telescoper for~$h_2$.  The converse can be proved by the same argument.
\end{proof}

The following lemma is a $q$-analogue of Theorem~7 in~\citep{AbramovPetkovsek2002a}.
\begin{lemma} \label{LM:split}
Let~$p$ be an irreducible polynomial in~$k[x, y]$. Then

\smallskip
\begin{enumerate}
  \item [(i)] If there exist~$i, j\in \bZ$ with~$i\neq 0$ and~$c\in k$ such that~$\si_x^i\tau_y^j(p) = cp$, then~$p\in k[y]$.
  \item [(ii)] If there exist~$i, j\in \bZ$ with~$j\neq 0$ and~$c\in k$ such that~$\si_x^i\tau_y^j(p) = cp$, then~$p=\lambda y$ for some~$\lambda \in k$ or~$p\in k[x]$.
\end{enumerate}
\end{lemma}
\begin{proof}
For the first assertion, assume $\si_x^i\tau_y^j(p) = cp$ for $i,j \in
\bZ$ and~$c \in k$, with~$i \neq 0$.  We consider two cases.  First,
if $j=0$ or~$\deg_y(p)=0$, then $\si_x^i(p) = cp$, so that $c=1$ by
comparing the leading coefficients.  Upon setting~$K = k(y)$ and
applying Lemma~2 of~\citep{AbramovPetkovsek2002a}, we conclude that~$p \in
K$, and then~$p \in k[y]$.
For the second case, we assume
$d=\deg_y(p)>0$ and~$j\neq 0$.  Write~$p = p_d(x) y^d + \cdots +
p_0(x)$, where~$p_0, \ldots, p_d \in k[x]$ and~$p_d\neq 0$. Then
\[\si_x^i \tau_y^j(p) = p_d(x+i)q^{jd}y^d + \cdots + p_0(x+i).  \]
From the equality~$\si_x^i\tau_y^j(p) = cp$, we get
\[p_\ell(x+i)q^{j\ell} = c p_\ell(x) \quad \text{for all~$\ell$ with~$0\leq \ell \leq d$.}\]
Assume that~$\ell$ is an integer in~$\{0, \ldots, d\}$ such
that~$p_\ell(x) \neq 0$.  Then~$c = q^{j\ell}$ by the above equation.
This implies that~$p_\ell(x+i) = p_\ell(x)$. By Lemma~2
in~\citep{AbramovPetkovsek2002a}, $p_\ell\in k$ for all~$\ell$
with~$p_\ell(x) \neq 0$. Thus,~$p\in k[y]$ again, and the first
assertion holds.

To prove the second assertion, we first show that the case~$i \neq 0$
reduces to the case~$i = 0$.  Indeed, if $\si_x^i\tau_y^j(p) = cp$ for
$i,j \in \bZ$ and~$c \in k$, and if~$i,j \neq 0$, the first assertion
of the lemma implies that~$p \in k[y]$.  As a consequence,
$\si_x^0\tau_y^j(p) = cp$~as well.  So we now focus to the special
case~$i = 0$, that is~$\tau_y^j(p)=cp$.
Suppose that there are~$d_1$ and~$d_2$ in~$\bN$ with~$d_1> d_2$ such that
\[p = p_{d_1} y^{d_1} + p_{d_2} y^{d_2} + \, \text{terms of lower degree in~$y$,}\]
where~$p_{d_1}, p_{d_2}\in k[x]^*$. Applying~$\tau_y^j$ to~$p$ yields
\[\tau_y^j(p) = p_{d_1} q^{jd_1}y^{d_1} + p_{d_2} q^{jd_2}y^{d_2} + \, \text{lower terms in~$y$.}\]
The equality~$\tau_y^j(p)=cp$ implies that~$c=q^{jd_1}$ and~$c=q^{jd_2}$. Hence, $q^{j(d_1-d_2)}=1$, a contradiction to the assumption
that~$q$ is not a root of unity. So~$p=\lambda y^s$ for some~$\lambda$ in~$k[x]$ and~$s$ in~$\bN$.
Since $p$~is irreducible, $s$ is equal to~$0$ or~$1$, and if~$s=1$,
then~$p \in k$, again because it is irreducible.  Therefore, the
second assertion holds.
\end{proof}

The following lemma is an analogue for the difference-$q$-difference
case of Theorem 9 in~\citep{Abramov2003}.
\begin{lemma} \label{LM:dqd}
Let~$f =a/b \in k(x, y)^*$ with~$\gcd(a,b)=1$, let~$b$ be non-split,
and let~$L$ be in~$k(x, y)\langle \pa_x\rangle$ whose coefficients are all split.
Then
\begin{enumerate}
  \item [(i)] If~$b$ is shift-free~\wrt~$\si_y$, and~$\pa_x = T_x$, then the non-split part of
the denominator of~$L(f)$ is not~$\si_y$-spread.
  \item [(ii)] If~$b$ is $q$-shift-free~\wrt~$\tau_y$, and~$\pa_x = S_x$, then the non-split part of
the denominator of~$L(f)$ is not~$\tau_y$-spread.
\end{enumerate}
\end{lemma}

\begin{proof} Let~$L = \sum_{i=0}^\rho u_i \pa_x^i$, where~$u_0, u_1, \ldots, u_\rho \in k(x,y)$
are split and~$u_\rho$ is nonzero.

First, assume that~$b$ is shift-free with respect to~$\si_y$, and that~$\pa_x = T_x$.
The first assertion clearly holds if~$\rho = 0$, because a nontrivial and shift-free polynomial with respect to~$\si_y$
is not $\si_y$-spread. So it suffices to consider the case in which~$\rho>0$.

Applying~$L$ to~$f$, we have that
\begin{equation} \label{EQ:applyq}
L(f) = \sum_{i=0}^\rho u_i \tau_x^i\left(\frac{a}{b} \right).
\end{equation}
Denote by~$B$ the non-split part of the denominator of~$L(f)$, which belongs to~$k[x,y]$.
If~$B$ is in~$k$, then the first assertion evidently holds. Assume further that~$B$ has positive
degrees in both~$x$ and~$y$. By~\eqref{EQ:applyq}, $b$ must have a nontrivial irreducible factor,
which is non-split.
Let~$p$ be such a factor.
Furthermore, we may assume without loss of generality that~$\tau_x^\mu(p)\nmid b$ for any positive integer~$\mu$.
Then~$\tau_x^\rho(p) \mid B$, because all the~$u_i$'s are split.

Suppose that~$B$ is~$\si_y$-spread. Then there exists~$j_0 \in \bZ^*$ such that~$\si_y^{j_0}\tau_x^\rho(p)\mid B$.
By~\eqref{EQ:applyq},~$B$ divides~$\prod_{i=0}^\rho \tau_x^i(b)$. By the above two divisibilities,
there exists~$\ell_0$ in~$\{0, \ldots, \rho\}$ such that~$\si_y^{j_0}\tau_x^\rho(p)\mid \tau_x^{\ell_0}(b)$.
Note that~$\ell_0 \neq \rho$, for otherwise, both~$\si_y^{j_0}(p)$ and~$p$ would divide~$b$,
a contradiction to the assumption that~$b$ is shift-free.
Therefore,~$\ell_0 < \rho$. Since~$\si_y^{j_0}\tau_x^{\rho-\ell_0}(p)\mid b$,
there exists a non-negative integer~$i_0$ such that
\[\si_y^{j_0}\tau_x^{\rho-\ell_0+i_0}(p)\mid b \quad \text{but} \quad \si_y^{j_0}\tau_x^{\rho-\ell_0+i_0+\mu}(p)\nmid b\]
for any positive integer~$\mu$.
It follows from~$\rho-\ell_0>0$ that~$\rho - \ell_0 + i_0 > 0$.

Repeating the above process for the irreducible factor~$\si_y^{j_0}\tau_x^{\rho-\ell_0+i_0}(p)$,
we can find that~$j_1 \in \bZ^*$,~$i_1 \in \bN$ and~$\ell_1 \in \{0, \ldots, \rho-1\}$
such that~$\rho -\ell_1 + i_1 > 0$ and
\[\si_y^{j_0+j_1}\tau_x^{(\rho-\ell_0+i_0) + (\rho-\ell_0+i_1)}(p) \mid b \quad \text{but} \quad
\si_y^{j_0+j_1}\tau_x^{(\rho-\ell_0+i_0) + (\rho-\ell_1+i_1) + \mu}(p) \nmid b\]
for any positive integer~$\mu$. Continuing this process yields a sequences of irreducible factors of~$b$. Since~$b$ has only finitely many irreducible
factors, there exist~$m, n\in \bN$ with~$n<m$ such that
\[ \si_y^{j_0+ \cdots + j_n}\tau_x^{(\rho-\ell_0+i_0) + \cdots + (\rho-\ell_n+i_n)}(p)=
c  \si_y^{j_0+ \cdots + j_m}\tau_x^{(\rho-\ell_0+i_0) + \cdots + (\rho-\ell_m+i_m)}(p)\]
for some~$c\in k$. This implies that
\[\si_y^{-j_{n+1}-\cdots - j_m}\tau_x^{-(\rho-\ell_{n+1}+i_{n+1}) - \cdots - (\rho-\ell_m+i_m)}(p) = cp.\]
Note that~$(\rho-\ell_{n+1}+i_{n+1}) + \cdots + (\rho-\ell_m+i_m) \neq
0$. By Lemma~\ref{LM:split}~(ii), $p$~is either in~$k[x]$ or is equal
to~$\lambda y$ for some~$\lambda \in k$, a contradiction to the
assumption that~$p$ is non-split.

The second assertion can be proved similarly according to Lemma~\ref{LM:split}~(i).
\end{proof}

We are ready to prove that a telescopable term is the sum of an exact term and a proper one.
\begin{theorem}\label{THM:chartele}
Let~$h$ be a mixed term over~$\mfield$.
Let
$$h = \Delta_y(h_1) + h_2$$
be an additive decomposition with respect to~$\pa_y$.
If~$h$ has a telescoper of type~$(\pa_x, \pa_y)$, then~$h_2$ is either zero or proper.
\end{theorem}
\begin{proof}
By the definition of additive decompositions, there exist~$r, u, v \in k(x,y)$
such that~$h_2 \in r \cH(u,v)$, where $r$, $u$ and~$v$ satisfy the properties described in
Definition~\ref{DEF:additive}. Clearly, we may assume that~$h_2$ is nonzero.
Assume further that~$h$ has a telescoper~$L$ of type~$(\pa_x, \pa_y)$.
Then~$L$ is also a telescoper of~$h_2$ by Lemma~\ref{LM:ns}.
Our goal is to show that~$h_2$ is proper.

First, we let~$\eta_y=\delta_y$. By Property~$(a)$ in Definition~\ref{DEF:additive},~$\den(r)$
is squarefree with respect to~$y$, and~$v$ is~$\delta_y$-reduced.
By Lemma~\ref{LM:reduced}~(i), both~$u$ and~$\den(v)$ are split.
Therefore,~$h_2$ is proper if and only if~$\den(r)$ is split by Lemma~\ref{LM:den}~(iv) and~(v).
So it  suffices to show that~$\den(r)$ is split.

\smallskip \noindent
{\em Case 1.1.} Let~$\xi_x = \si_x$. Then we may write
\[ L= \sum_{i=0}^\rho e_iS_x^i \in k(x)\langle S_x\rangle, \quad \mbox{where~$e_0, \ldots, e_\rho \in k(x)$ and~$e_\rho\neq 0$}.
\]
As in Example~\ref{EX:apply}, we have
\[L(h_2) =\underbrace{\left(\sum_{i=0}^\rho u_i \si_x^i\left(r\right)\right)}_{f} g,\]
where~$g \in \cH(u, v) \setminus \{0\}$, ~$u_i = e_i \prod_{j=0}^{i-1} \si_x^j(u)$, and~$u_\rho$ is nonzero.

Suppose that~$\den(r)$ is non-split. Then there exists a non-split irreducible polynomial~$p \in k[x,y]$
such that~$p \mid \den(r)$ and~$\si_x^i(p) \nmid \den(r)$ for all~$i>0$. Moreover,~$p$ has
multiplicity one in~$\den(r)$, because~$\den(r)$ is squarefree. Since all the~$u_i$'s are split,~$\si_x^\rho(p)$
is an irreducible factor of~$\den(f)$ with multiplicity one.
In particular,~$f$ is nonzero.
It follows that the non-split part of~$\den(f)$
is not $\delta_y$-spread.
Since~$L$ is a telescoper of~$h_2$ whose type is~$(S_x, D_y)$, $fg$ is an exact term with respect to~$D_y$.
Note that~$g \in \cH(u, v)$ and~$\den(v)$ is split. We can apply Proposition~\ref{PROP:sum}~(i) to~$L(h_2)$,
which implies that the non-split part of~$\den(f)$ is~$\delta_y$-spread, a contradiction.
Thus,~$\den(r)$ is split.

\smallskip \noindent
{\em Case 1.2.} Let~$\xi_x = \tau_x$. The same argument used in Case~1.1 shows that~$h_2$ is proper.

\medskip
Second, we let~$\eta_y=\si_y$. By Property~$(b)$ in Definition~\ref{DEF:additive},~$\den(r)$
is shift-free with respect to~$\si_y$, and~$v$ is~$\si_y$-reduced.
Both~$\den(u)$ and~$v$ are split if~$\xi_x = \de_x$ by Lemma~\ref{LM:reduced}~(2a).
Both~$u$ and~$v$ are split if~$\xi_x =\tau_x$ by Lemma~\ref{LM:reduced}~(2b).
Therefore,~$h_2$ is proper if and only if~$\den(r)$ is split by Lemma~\ref{LM:den}~(i) and~(vi).
So it  suffices to show that~$\den(r)$ is split.

\smallskip \noindent
{\em Case 2.1.} Let~$\xi_x = \de_x$. Then we may write
\[ L= \sum_{i=0}^\rho e_iD_x^i \in k(x)\langle D_x\rangle, \quad \mbox{where~$e_0, \ldots, e_\rho \in k(x)$ and~$e_\rho\neq 0$}. \]
As in Example~\ref{EX:apply}, we have
\[L(h_2) = \underbrace{\left(e_\rho \de_x^\rho(r) + \sum_{i=0}^{\rho-1} u_i \de_x^i(r)\right)}_{f} g,\]
where~$g \in \cH(u, v) \setminus \{0\}$, $u_i \in k(x, y)$, and~$\den(u_i)$
is split for all~$i$ with~$0 \le i \le \rho-1$.

Suppose that~$\den(r)$ is non-split. Then there exists a non-split irreducible polynomial~$p \in k[x,y]$
such that~$p \mid \den(r)$. Assume that~$m$ is the multiplicity of~$p$ in~$\den(r)$. Then
the multiplicity of~$p$ in~$\den(\delta^i(r))$ is equal to~$m+i$. In particular, its multiplicity
in~$\den\left(e_\rho \delta^\rho(r)\right)$ is equal to~$m+\rho$, which is also the multiplicity of~$p$ in~$\den(f)$,
because all the~$\den(u_i)$ are split.
It follows that~$p$ is an irreducible factor of~$\den(f)$. In particular,~$f$ is nonzero.
Since~$L$ is a telescoper of~$h_2$ whose type is~$(D_x, S_y)$, $fg$ is an exact term with respect to~$S_y$.
Note that~$g \in \cH(u, v)$ and that~$v$ is split. We can apply Proposition~\ref{PROP:sum}~(ii) to~$fg$,
which concludes that the non-split part of~$\den(f)$ is~$\si_y$-spread. Thus, there exists
a nonzero integer~$j$ such that~$\si_y^j(p)$ is an irreducible factor of~$\den(f)$.
From the definitions of~$f$ and splitness of the~$\den(u_i)$,
it follows that~$\si_y^j(p)$ is also an irreducible factor of~$\den(r)$, a contradiction to the fact
that~$\den(r)$ is shift-free. Hence,~$\den(r)$ is split.

\smallskip \noindent
{\em Case 2.2.} Let~$\xi_x = \tau_x$. Then we may write~$L= \sum_{i=0}^\rho e_iT_x^i \in k(x)\langle T_x \rangle$
\[ L= \sum_{i=0}^\rho e_iT_x^i \in k(x)\langle T_x \rangle, \quad \mbox{where~$e_0, \ldots, e_\rho \in k(x)$ and~$e_\rho\neq 0$}. \]
As in Example~\ref{EX:apply}, we have
\[L(h_2) = \underbrace{\left(\sum_{i=0}^\rho u_i \tau_x^i(r)\right)}_{f} g, \]
where~$g \in \cH(u, v) \setminus \{0\}$ and $u_i \in k(x,y)$ is split for all~$i$ with~$0 \le i \le \rho$.

Suppose that~$\den(r)$ is non-split.
By Lemma~\ref{LM:dqd}~(i),~$\den(f)$ is not $\si_y$-spread.
Since~$L$ is a telescoper of~$h_2$ whose type is~$(T_x, S_y)$, $fg$ is an exact term with respect to~$S_y$.
Note that~$g \in \cH(u, v)$ and~$v$ is split. Applying Proposition~\ref{PROP:sum}~(ii) to~$fg$, we see
that~$\den(f)$ is $\si_y$-spread,
a contradiction. Hence,~$\den(r)$ is split.

\medskip
Third, we let~$\eta_y=\tau_y$. By Property~$(c)$ of Definition~\ref{DEF:additive},~$\den(r)$
is $q$-shift-free with respect to~$y$, and~$v$ is~$\tau_y$-reduced.
Both~$\den(u)$ and~$v$ are split if~$\xi_x = \de_x$ by Lemma~\ref{LM:reduced}~(3a).
Both~$u$ and~$v$ are split if~$\xi_x =\tau_x$ by Lemma~\ref{LM:reduced}~(3b).
Therefore,~$h_2$ is proper if and only if~$\den(r)$ is split by Lemma~\ref{LM:den}~(ii) and~(iii).
So it suffices to show that~$\den(r)$ is split.

\smallskip \noindent
{\em Case 3.1.} Let~$\xi_x=\delta_x$. The proof is similar to that in Case~2.1, in which one
applies the third assertion of Proposition~\ref{PROP:sum} instead of the second one.

\smallskip \noindent
{\em Case 3.2.} Let~$\xi_x=\si_x$. The proof is similar that in Case~2.2, in which one
applies the second assertion of Lemma~\ref{LM:dqd} instead of the first one, and
the third assertion of Proposition~\ref{PROP:sum} instead of the second one.
\end{proof}

Combining Theorems~\ref{TH:fund} and~\ref{THM:chartele}, we obtain a criterion for
the existence of telescopers of a mixed term, which is the main result of this article.
\begin{theorem} \label{TH:criterion}
Let~$h$ be a mixed term over~$\mfield$. Assume that
$$h = \Delta_y(h_1) + h_2$$
is an additive decomposition of~$h$. Then~$h$ has a telescoper
of type~$(\pa_x, \pa_y)$ if and only if~$h_2$ is either zero or a proper mixed term.
\end{theorem}

\section{Algorithms and examples} \label{SECT:algo}

For a given mixed term, we can decide the existence of telescopers by Theorem~\ref{TH:criterion}.
First, we use the algorithms in~\citep{AbramovPetkovsek2002b}, \citep{GeddesLeLi2004}, and~\citep{ChenHouMu2005} to perform
the respective additive decompositions. Second,
we test whether the denominator of the rational part in the non-exact component is split or not by Remark~\ref{RE:split}.
The decision procedure is given in Figure~\ref{fig:algo}.

\begin{center}
\begin{figure}
\framebox[13.2cm]{
\begin{minipage}{13.1cm} \rule[0.3cm]{0cm}{0cm}
{{Algorithm \textsf{IsTelescopable}}

\noindent \quad{\textbf{Input}}: a mixed term $h\in \cH(a, b)$ over~$\mfield$.

\noindent \quad{\textbf{Output}}:~true, if~$h$ has a telescoper of type~$(\xi_x, \eta_y)$; false, otherwise.

\medskip
\begin{enumerate}
\item Compute an additive decomposition of~$h$ with respect to~$\eta_y$ and get
\[h = \Delta_y(h_1) + rg, \]
where~$g \in \cH(u,v)$ is as given
by~(\ref{EQ:add}--\ref{eq:additive-decomposition}) in
Definition~\ref{DEF:additive}.
\smallskip
\item Compute the primitive part~$p$ of~$\den(r)$ with respect to~$y$.
\smallskip
\item If~$p$ is in~$k[y]$, then return true, otherwise, return false.
\end{enumerate}
}
\end{minipage}
}
\caption{Algorithms for deciding the existence of telescopers.} \label{fig:algo}
\end{figure}
\end{center}

\begin{example}\label{EXAM:notredundant}
It is possible that a mixed term has a telescoper of type~$(S_x, D_y)$ but no
telescoper of type~$(D_x, S_y)$ or~$(D_x, T_y)$. Consider the rational function
\[ h = \frac{1}{(x+y)^2}.\]
Applying Hermite reduction to~$h$ with respect to~$\de_y$ yields
\[h = D_y\left(\frac{-1}{x+y}\right),\]
which implies that~$1$ is a telescoper of type~$(S_x, D_y)$
for~$h$.
Note that~$h = \Delta_y(0) + h$ is an additive decomposition when~$\eta_y=\si_y$
or~$\eta_y=\tau_y$, because~$\den(h)=(x+y)^2$ is both shift-free and $q$-shift-free with
respect to~$y$. But~$h$ is not proper, because~$x+y$ is not split. Hence, $h$~has no telescoper of type~$(D_x, S_y)$ and~$(D_x, T_y)$.
Similarly, consider the rational function
\[h = \frac{1}{(x+y)(x+y+1)}.\]
Since~$h=(S_y-1)(-1/(x+y))$, $1$ is a telescoper of type~$(D_x, S_y)$ for~$h$. However,
$h$ has no telescoper of type~$(S_x, D_y)$ or~$(T_x, D_y)$, because~$(x+y)(x+y+1)$ is squarefree with respect to~$\de_y$ and
it is not split.
\end{example}

As we mentioned before, properness is only a sufficient condition for the existence of telescopers.
The following two examples illustrate this fact.
\begin{example}\label{EXAM:notproper}
Consider the mixed term over~$\left(k(x, y), (\si_x, \de_y)\right)$
\begin{equation}\label{eq:notproper}
h = \frac{-y+2xy+2x^2}{(x+y)^2 x} \cdot y^x \cdot e^{-y}.
\end{equation}
In the structural decomposition~\eqref{eq:notproper} given by $\alpha
= -1$, $\beta = y$, $\gamma = 1$, and $f = h / (y^x e^{-y})$, the
denominator~$(x+y)^2$ is not split.  By case~(iv) in
Lemma~\ref{LM:den}, it follows that~$h$ is not proper.  But $h$~has a
telescoper of type~$(S_x, D_y)$ since it can be decomposed into
\[h = D_y \left(\frac{1}{x+y}  \cdot y^x \cdot e^{-y}\right) +  y^{x-1} \cdot e^{-y},\]
where~$y^{x-1} \cdot e^{-y}$ is proper, because the rational function in
the corresponding structural decomposition is~$1$, and therefore its denominator is split.
\end{example}

\begin{example} \label{EXAM:qnotproper}
Consider the mixed term over~$\left(k(x, y), (\tau_x, \de_y)\right)$
\[h = \frac{y^2 + xy -x}{(x+y)^2 x} \cdot e^{-y}, \]
which is not proper, because, in the structural decomposition given by
$\alpha = -1$, $\beta = 1$, $f = h / e^{-y}$, the
denominator~$(x+y)^2$ is not split. But it has a telescoper of
type~$(T_x, D_y)$ since~$h$ can be decomposed into
\[h = D_y \left(\frac{1}{x+y}\cdot e^{-y}\right) + \frac{e^{-y}}{x},\]
where~$e^{-y}/x$ is proper, because the rational function in
the corresponding structural decomposition is~$1/x$, and therefore its denominator is split.
\end{example}

The last example presents another application of Theorem~\ref{TH:criterion}.
\begin{example}\label{EX:transcendental}
Let
\[f = \frac{1}{y^2 - x}.\]
Note that the denominator of~$f$ is non-split and shift-free with respect to~$\si_y$. By
Theorem~\ref{TH:criterion}, there is no linear differential
operator~$L(x, D_x) \in k(x)\langle D_x\rangle$ and $g \in k(x,y)$ such that $L(x,D_x)(f) = \Delta_y(g)$,
which, together with Proposition~3.1 in~\citep{Hardouin2008} and the descent argument similar to that
given in the proof of Corollary~3.2 in~\citep{Hardouin2008} (or Section 1.2.1
of~\citep{DH2012}), implies that the sum
\[F(x, y) := \sum_{i=1}^{y -1}\frac{1}{{i}^2-x}\quad   \mbox { (satisfying~$S_y(F) - F = f$) }\]
satisfies no polynomial differential equation~$P(x, y, F, D_x(F),D_x^2(F), \ldots ) = 0$.
\end{example}

\bibliographystyle{elsart-harv}

\end{document}